\documentclass{article}
\usepackage{booktabs}
\usepackage{amsmath,amsthm}
\usepackage{amssymb}
\usepackage{color}
\usepackage{geometry}
\usepackage{booktabs}
\usepackage{tabularx,lipsum,environ}
\usepackage{nicematrix}
\usepackage{multirow}
\usepackage{graphicx}
\usepackage[hidelinks]{hyperref}
\usepackage{enumitem}
\setlist[itemize]{noitemsep, topsep=0pt}
\usepackage{subcaption} 
\usepackage{setspace}
\usepackage{nicefrac}
\usepackage{tikz}
\usetikzlibrary{arrows.meta}
\newcommand{\XC}{{\sc exact 3-cover}}
\newtheorem{definition}{Definition}
\newtheorem{observation}{Observation}
\newtheorem{theorem}{Theorem}
\newtheorem{example}{Example}
\newtheorem{lemma}{Lemma}
\newtheorem{claim}{Claim}
\newtheorem{remark}{Remark}

\newcommand{\namelong}{Pareto-Improving Random Matchings under Ex-post Stability}
\newcommand{\name}{PIRMES}
\newcommand{\myprob}{{\sc IsConstrained-SD-Eff}}

\newcommand{\acset}{A}
\definecolor{zold}{rgb}{0,0.545,0}
\newcommand{\gnote}{\textbf{Gergely:}\textcolor{zold}}
\usepackage{tcolorbox} 

\usepackage[numbers]{natbib}
\newcommand{\capac}{c}
\newcommand{\M}{\mathcal{M}}
\newcommand{\tM}{\tilde{\M}}

\newcommand{\haris}[1]{\textcolor{red}{\textbf{Haris says:} #1}}

\newcommand{\rank}{\mathrm{rank}}
\newcommand{\peter}[1]{\textcolor{magenta}{\textbf{Peter says:} #1}}

\newcommand{\comsmti}{{\sc Com-SMTI}}
\newcommand{\tom}[1]{\textcolor{blue}{\textbf{Tom says:} #1}}

\newcommand{\set}{\mathcal{S}}
\newtcolorbox{problembox}[1][]{%
  title=#1,
  colframe=black,
  coltitle=white,
  colback=gray!10,
  boxrule=0.5mm,
  sharp corners,
  before skip=10pt,
  after skip=10pt,
  fonttitle=\bfseries,
}

\newcommand{\probleminput}{\textbf{Input:} }
\newcommand{\problemquestion}{\textbf{Question:} }

\newcommand{\bestgreedyEE}{{\sc Max-Constrained-Improve-Greedy}}
\newcommand{\bestimproveEE}{{\sc Max-Constrained-Improve}}
\newcommand{\avgrank}{\mathrm{avgrank}}

\title{Smart Lotteries in School Choice:\\ Ex-ante Pareto-Improvement with Ex-post Stability}
\author{
Haris Aziz\textsuperscript{1},
Péter Biró\textsuperscript{2},
Gergely Csáji\textsuperscript{3.4},
Tom Demeulemeester\textsuperscript{5}\\
\\
\textsuperscript{1}UNSW Sydney, Australia\\
\textsuperscript{2}ELTE Centre for Economic and Regional Studies, Budapest, Hungary\\
\textsuperscript{3}HUN-REN Institute for Computer Science and Control, Budapest, Hungary\\
\textsuperscript{4}Budapest University of Technology and Economics, Budapest, Hungary\\
\textsuperscript{5}Department of Quantitative Economics, Maastricht University, The Netherlands
}
\date{}
\begin{document}
\maketitle

\begin{abstract}
In school choice, students typically have strict preferences over the schools while the schools have coarse priorities over the students based on their distance and their enrolled siblings. The outcome of a centralized admission mechanism is then usually obtained by the Deferred Acceptance (DA) algorithm with random tie-breaking. Therefore, every outcome of this mechanism that occurs with a positive probability is a stable solution for the coarse priorities. This implies a probabilistic assignment, specifying the admission probability for each student-school pair. In this paper, we propose a new efficiency-improving stable `smart lottery' mechanism. We aim to improve the probabilistic assignment ex-ante in a stochastic dominance sense, while ensuring that the improved random matching is still ex-post stable, meaning that it can be decomposed into stable matchings regarding the original coarse priorities. Therefore, this smart lottery mechanism can provide a clear Pareto-improvement in expectation for any cardinal utilities compared to the standard DA with lottery solution, without sacrificing the stability of the final outcome. We show that although the underlying computational problem is NP-hard, we can solve the problem in practice by using column generation. Our computational results show that the welfare gains by our mechanism are substantially larger than the expected gains by standard methods that realize efficiency improvements after ties have already been broken. Additionally, we illustrate for real-world school choice data of Ghent (Belgium) that, while not theoretically guaranteed, our mechanism can Pareto-improve upon methods that improve upon DA but violate ex-post stability, such as Efficiency-Adjusted Deferred Acceptance.
\end{abstract}

\maketitle

\section{Introduction}





We consider the school choice problem, that is a two-sided many-to-one matching problem, in which students are matched to schools based on the preferences of the students and their priorities at the schools ~\citep{AbSo03b}. The underlying model originates from the Gale-Shapley college admission problem \citep{gale1962college}, where both the students and the colleges have preferences over the opposite sets. The important distinction between the two models is that in school choice we only consider the students as strategic agents, and we care about their welfare. Nevertheless, the same solution concept and algorithm is used in most applications, the \textit{Deferred Acceptance (DA)} algorithm by Gale and Shapley. The DA mechanism provides a \emph{stable matching} for college admission, and a so-called \emph{justified-envy-free}, \emph{non-wasteful} matching in the school choice context. These notions are equivalent, essentially meaning that every rejection is fairly made, as the given school filled its quota with higher priority students than the rejected applicant. See a historical overview on the DA in~\cite{Roth08a}, and several books on the topic of stable matchings by economists~\citep{RoSo90a,haeringer2018market}, computer scientists \citep{Manl13a}, and multi-disciplinary researchers~\citep{echenique_etal2023}.

When the priorities are coarse, the ties are typically broken by lotteries, and then DA is applied afterwards. The extent to which the priorities are weak differ across school choice applications. In the US, the typical approach is to consider four basic categories based on the catchment areas and siblings, as in the New York high school admissions \citep{Atila_etal2005NY}. However, in some other applications, such as for high schools in Amsterdam, no priorities are used and only a single lottery is conducted \citep{RuijsOosterbeek2019}, so the mechanism is a so-called \emph{Random Serial Dictatorship} (RSD). The lottery can be a single lottery applied for all schools, or multiple lotteries applied separately for each school, see a recent study on this question \citep{Arnosti2023lottery}. Note that although we focus our attention on school choice in this paper, the use of lotteries to break coarse priorities also occurs in many other applications, such as university admission, resident allocation, or course allocation (see Section~\ref{sec:related}). Therefore, our new smart lottery mechanism can potentially be used more widely.

The DA mechanism is strategy-proof, and provides a student-optimal matching among the stable matchings. However, even for strict priorities, it is not Pareto-efficient \citep{AbSo03b}. When the priorities are coarse, and the solution is obtained by DA with lottery, then the matching may not even be \emph{constrained-efficient}, which means that a Pareto-improvement can be possible without violating the original coarse priorities (but only the priorities by the lotteries). \cite{erdil2008s} have proposed a deterministic mechanism for making stable Pareto-improvements, upon the solution by the DA with lottery, to obtain a constrained-efficient matching by their so-called \emph{stable improvement cycles algorithm}. They demonstrated on some New York high school matching datasets that their approach could Pareto-improve the DA with lottery matching for 2-3\% of the students (without hurting the others, and by only violating the lottery-based priorities). Nevertheless, this proposal was rejected, partly due to the unavoidable manipulability of the mechanism \citep{Atila_etal2009NY}, but perhaps also because the violation of lottery-priorities can be problematic from a legal perspective in the US applications. In our paper, we use the Erdil-Ergin (EE) solution as a benchmark for comparison with the outcome of our new mechanism. Note also that as a new theoretical result related to the concept of Erdil and Ergin, we show that computing a lowest average-rank constrained Pareto-improvement is an NP-hard problem even for this deterministic case.  

Our new proposal is a smart lottery mechanism based on ex-ante Pareto-improvement, while keeping the final matchings stable. To describe the idea,  we first have to introduce random matchings. When a lottery is applied in some matching application, the different outcomes (matchings) will occur with different probabilities. A \emph{random matching} (or \emph{probabilistic assignment}) describes the probability of each student-school pair to get realized in the final outcome. A random matching $p$
is ex-ante Pareto-dominated by another random matching $q$ in a \emph{stochastic dominance} (sd) sense if for every student-school pair the probability of this student getting that school or a more preferred one can only increase from $p$ to $q$, with a strict increase for at least one pair. Such an improvement is called \emph{sd-improvement}, and if a random matching cannot be sd-improved then it is called \emph{sd-efficient}.

In their seminal paper, \cite{BoMo01a} have shown that for standard object allocation problems (with no priorities), the random matching obtained by the \textit{Random Serial Dictatorship (RSD)} rule is not necessarily sd-efficient, and they proposed an alternative mechanism, the so-called \emph{Probabilistic Serial} mechanism, to obtain an sd-efficient random matching. 

Recall that RSD is used for school choice in Amsterdam \citep{RuijsOosterbeek2019}, and also for resident allocation in Israel. However, a remarkable redesign has been implemented in the latter application, documented in detail by \citet{bronfman2018redesigning}. The new proposal is a two-step mechanism. In the first step, the random matching by RSD is Pareto-improved for estimated utilities. In the second step, a smart lottery is conducted, where the random matching is decomposed into a convex combination of desirable matchings, where the couples are getting positions close to each other with high probability.  


Our approach follows a similar spirit, but differs in three main aspects. First, since we have a school choice setting with coarse priorities, as opposed to no priorities, stability plays a key role in our paper. As such, we improve upon the random matching obtained by DA with lotteries, instead of RSD. Secondly, in the first step of the mechanism we are a bit more conservative with the Pareto-improvement, and we only look for sd-improvements, i.e., a guaranteed Pareto-improvement for any underlying cardinal utilities, and not only for estimated utilities. Finally, we do not have couples in our setting, but we want to ensure that the final outcome is stable with respect to the coarse priorities.

If a random matching can be decomposed into stable matchings for coarse priorities, then it is called \emph{ex-post stable}, as defined by \citet{KestenUnver2015}. Therefore, in our mechanism we make the ex-ante sd-improvement in the first phase such that the random matching is ex-post stable, and then it can be decomposed into stable matchings for the smart lottery.

Formally, our \textit{\namelong} (\name) mechanism works as follows.  


\begin{tcolorbox}[title= \textit{\namelong} (\name), title filled, before upper={\linespread{1.4}\selectfont}]
	\begin{itemize}
\item STEP 0: Compute a random matching $p$, e.g., with DA and uniform tie-breaking
\item STEP 1: Compute another random matching $q$ that sd-improves upon $p$, such that the improvement is the largest possible in terms of average ranks and $q$ is ex-post stable. If no such ex-post stable random matching $q$ can be found (because $p$ was not ex-post stable), return $p$.
\item STEP 2: Decompose $q$ into a convex combination of stable matchings and conduct the smart lottery.
\end{itemize}
\end{tcolorbox}


Our mechanism can be applied to any random matching $p$. In this paper, we will mainly evaluate our mechanism by improving upon DA with uniform tie-breaking, and we refer to this mechanism as DA-\name. This choice is motivated by DA being a standard method in practical school choice problems, and by the guaranteed stability of the returned matching, unlike methods such as, e.g., Efficiency-Adjusted Deferred Acceptance (EADA) \citep{kesten2010school}. 
Nevertheless, we evaluate our mechanism when applied to other random matchings in Section \ref{sec:comp_exper}.

The main conceptual difference with the Erdil-Ergin approach is that we make the Pareto-improvement ex-ante, and we conduct the smart lottery only at the end of the mechanism.

We demonstrate the mechanism with a simple example.

\begin{example}
\label{ex:basic}

Consider the following problem instance.
We have four students $1,2,3,4$ and four schools $s_1,s_2,s_3,s_4$ with unit capacities, and with the following preferences and priorities.

    \begin{center}
        \begin{minipage}{0.3\textwidth}
\centering
\begin{tabular}{rl}
      $1:$ & $s_1\succ s_3 \succ s_4 \succ s_2$ \\
      $2:$ & $s_1\succ s_4 \succ s_3\succ s_2$ \\
      $3:$ & $s_2\succ s_3\succ s_4 \succ s_1$ \\
      $4:$ & $s_2\succ s_4\succ s_3 \succ s_1$ 
      \end{tabular}
\end{minipage}
\hspace{0.01\textwidth}
\begin{minipage}{0.3\textwidth}
\centering
\begin{tabular}{rl}
      $s_1:$ & $[1,2] \succ 3\succ 4$ \\
      $s_2:$ & $[3,4]\succ 1 \succ 2$ \\
      $s_3:$ & $2\succ 4\succ [1,3]$ \\
      $s_4:$ & $1\succ 3 \succ [2,4]$
      \end{tabular}
\end{minipage}
\end{center}
Here, the notation $[i,j]$ for some agents $i,j$ means that they have the same priority at a school. Figure \ref{fig:main_ex} illustrates the first two preferences of the students, while all schools are indifferent between the students who rank that school first or second.

The possible weakly stable matchings that can arise from DA with uniform single tie-breaking and their probabilities are: 

\noindent
\begin{itemize}
    \item $M_1= \{ (1,s_1),(2,s_3),(3,s_2),(4,s_4)\}$ with probability $\nicefrac{1}{8}$,
    \item $M_2=\{ (1,s_1),(2,s_4),(3,s_2),(4,s_3)\}$ with probability $\nicefrac{1}{8}$,
    \item $M_3=\{ (1,s_1),(2,s_4),(3,s_3),(4,s_2)\}$ with probability $\nicefrac{1}{4}$,
    \item $M_4= \{ (1,s_3),(2,s_1),(3,s_2),(4,s_4)\}$ with probability $\nicefrac{1}{4}$,
    \item $M_5= \{ (1,s_3),(2,s_1),(3,s_4),(4,s_2)\}$ with probability $\nicefrac{1}{8}$,
    \item $M_6=\{ (1,s_4),(2,s_1),(3,s_3),(4,s_2)\}$ with probability $\nicefrac{1}{8}$.
\end{itemize}

When we take the average probabilities for each student-school pair, we obtain the following random matching $p$:
\begin{itemize}
    \item \makebox[2.6cm][l]{first choices:}      
    $p(1,s_1)=p(2,s_1)=p(3,s_2)=p(4,s_2)=\nicefrac{1}{2}$,
    \item \makebox[2.6cm][l]{second choices:}
$p(1,s_3)=p(2,s_4)=p(3,s_3)=p(4,s_4)=\nicefrac{3}{8}$,
    \item \makebox[2.6cm][l]{third choices:} $p(1,s_4)=p(2,s_3)=p(3,s_4)=p(4,s_3)=\nicefrac{1}{8}$.
\end{itemize}
\noindent However, consider the following random matching $q$ which sd-dominates $p$, and is therefore strictly better for all students:

\begin{itemize}
    \item \makebox[2.6cm][l]{first choices:} $q(1,s_1)=q(2,s_1)=q(3,s_2)=q(4,s_2)=\nicefrac{1}{2}$,
    \item \makebox[2.6cm][l]{second choices:} $q(1,s_3)=q(2,s_4)=q(3,s_3)=q(4,s_4)=\nicefrac{1}{2}$.
\end{itemize}

\begin{center}
\begin{figure}[!tb]
\centering
\begin{tikzpicture}[scale=0.7,
    every node/.style={draw, thick, minimum size=6.5mm, font = \footnotesize},
    circ/.style={circle},
    sq/.style={rectangle},
    empty/.style={draw= none},
    solid/.style={->, thick},
    dashedw/.style={->, thick, dashed}
]

\node[circ] (n1) at (-2, 3) {1};
\node[circ] (n2) at ( 2, 3) {2};
\node[circ] (n3) at (-2,-3) {3};
\node[circ] (n4) at ( 2,-3) {4};

\node[circ] (S1) at (0, 1) {$s_1$};
\node[circ] (S2) at (0,-1) {$s_2$};
\node[circ] (S3) at (-2,0) {$s_3$};
\node[circ] (S4) at ( 2,0) {$s_4$};

\draw[solid, bend left=15]  (n1) to node[midway, draw=none, fill=none, above] {1} (S1);
\draw[solid, bend right=15] (n2) to node[midway, draw=none, fill=none, above] {1} (S1);

\draw[solid, bend right=15] (n3) to node[midway, draw=none, fill=none, below] {1} (S2);
\draw[solid, bend left=15]  (n4) to node[midway, draw=none, fill=none, below] {1} (S2);

\draw[dashedw]  (n1) to node[midway, draw=none, fill=none, left] {2} (S3);
\draw[dashedw] (n3) to node[midway, draw=none, fill=none,left] {2} (S3);

\draw[dashedw] (n2) to node[midway, draw=none, fill=none, right] {2} (S4);
\draw[dashedw]  (n4) to node[midway, draw=none, fill=none, right] {2} (S4);

\end{tikzpicture}

\caption{\centering First two preferences of the students in Example \ref{ex:basic}.}
\label{fig:main_ex}
\end{figure}
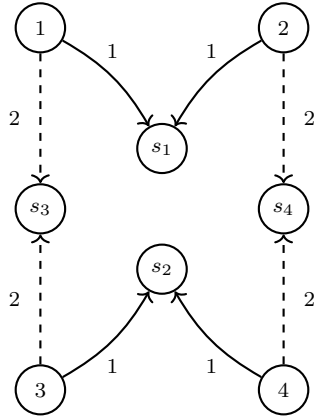
\end{center}

\vspace{-10mm}

Note that this random matching $q$ is exactly $\frac{1}{2}(M_3+M_4)$, so $q$ is ex-post stable. The underlying intuition behind this example is that by correlating the tie-breaking for the first choices, we can avoid competition for less preferred schools. If student 1 does not receive $s_1$, then allowing student 3 to receive $s_2$ will prevent students 1 and 3 from competing for their second choice $s_3$, for example. 

\end{example}


Regarding the implementation of our mechanism, we face computational challenges. \cite{aziz2026ex} proved recently that it is NP-complete to decide whether a given random matching is ex-post stable, and they provided an integer programming formulation to maximize the fraction of stable matchings in the decomposition of that random matching. Likewise, we will show that the main computational task in our mechanism is NP-hard, but we demonstrate that with sophisticated optimization methods the problem becomes solvable for realistic instances, and we evaluate the expected welfare gains with computational simulations conducted on generated and real datasets. 

\subsection*{Contributions}

Our main conceptual contribution is formalizing a smart lottery approach where an efficient random matching is computed which stochastically dominates a given random matching but maintains the requirement of being ex-post stable, thereby allowing it to be implemented by stable matchings. We formalize our approach via a class of mechanisms called \textit{\namelong} (\name). Within this class, we especially focus on DA-\name~ where we improve upon DA with uniform tie-breaking.
Our results are summarized below. 


\paragraph{Theoretical Results}

We present several complexity results for computational problems central to our study. Since our goal is to compute sd-improvements constrained to ex-post stability, a fundamental problem is checking whether a given ex-post stable random matching is constrained-sd-efficient. We prove that the problem is coNP-complete even if the students have preferences of length at most 3 and every capacity is 1 and the random matching is a uniform combination of student optimal stable matchings with respect to some tie-breakings.
    The result contrasts with the facts that (i) testing sd-efficiency of any given random matching is polynomial-time solvable and (ii) testing constrained efficiency of a deterministic weakly stable matching is also polynomial-time solvable~\citep{erdil2008s}.
    Regarding the latter, we show that finding a weakly stable matching $M'$ that Pareto-dominates a weakly stable matching $M$ and improves average rank by at least a given amount is NP-complete, hence it is computationally challenging to find ``improvement paths'' that lead to the best final matching. 
    
    We also show that given a matching $M$, it is NP-complete to decide if there is a weakly stable matching $N$, which Pareto-dominates $M$ for the students, or an ex-post stable random matching $p$ that sd-dominates $M$ for the students. Furthermore, we show it is also NP-hard to compute the uniform combination of all weakly stable matchings that arise as a student optimal stable matching with respect to some tie-breaking. 
    
    Finally, we show that no strategy-proof, ex-post stable, and constrained-sd-efficient random mechanism exists. Our complexity results motivate an operations research approach for our paradigm. 

    \paragraph{OR approach and Experiments} In light of these hardness results, we propose a column generation framework in Section \ref{sec:column_gen} to find the ex-post stable random matching of minimal average rank that sd-dominates a given random matching $p$. In essence, this framework will start by finding a lottery over a subset of the weakly stable matchings, and will continuously generate additional weakly stable matchings until a certificate of optimality is obtained. 
    
    Our simulations on generated and real-world data in Section \ref{sec:comp_exper} indicate that, without losing ex-post stability, our proposed solution is capable of finding random matchings with substantially lower average ranks than the standard approach by \citet{erdil2008s} (EE), which resolves inefficiencies in DA after ties have been broken. Our proposed solution achieves at least as many beneficiaries as EE, while obtaining substantially larger welfare gains for those improving students. For real data of an Estonian kindergarten application and for a particular priority structure, for example, our proposed solution realizes a sixfold increase, compared to EE, in the fraction of improving students upon DA, as well as a twelvefold increase in their average improvement in rank. 
    
    Lastly, for some structural properties of the preferences and priorities, our method manages to sd-dominate the expected outcome by existing mechanisms that remove inefficiencies in DA but do \textit{not} require stability, such as EADA \citep{kesten2010school}. Our computational experiments for real-world school choice data in Ghent (Belgium), for example, show that our methods can realize larger efficiency gains upon DA than EADA, while recovering the ex-post stability that was sacrificed by EADA to obtain these efficiency gains.

\section{Related Work}
\label{sec:related}


The theory of stable matchings has a long history, surveyed in several monographs~\citep{GuIr89a,Manl13a,RoSo90a,echenique_etal2023}. Foundational work on the stable matching polytope without ties~\citep{roth1993stable} and its linear programming connections~\citep{TeSe98a} provides important insights into random stable matchings. 
While the classical stable marriage problem assumes strict preferences~\citep{gale1962college}, the introduction of ties~\citep{IMS00a} significantly changed the landscape: although efficient algorithms exist for weak, strong, and super stability, many related problems become NP-hard when ties are allowed~\citep{Manlove_etal2002}. 

The issue of achieving stability and Pareto optimality is central to market design. 
In general, Pareto optimality and stability are incompatible~\citep{AbSo03b}. Therefore, the target is typically revised to finding efficient matchings within the space of stable matchings whether we consider random matchings or discrete matchings. When both sides have strict preferences / priorities, the Deferred Acceptance algorithm gives a matching which is Pareto optimal among stable matchings. However, this is not necessarily the case when the preferences and priorities allow for ties~\citep{erdil2008s}. 

\cite{erdil2008s} study two-sided matching markets in which agents may be indifferent between multiple partners. They characterize matchings which are ordinally efficient within the stable set and propose an algorithm that takes a stable matching as input and iteratively eliminates stable improvement cycles to reach a Pareto-efficient matching constrained to stable matchings. 
They highlight that commonly used procedures, such as Deferred Acceptance with arbitrary tie-breaking, can produce stable outcomes that are Pareto-dominated by other stable matchings. As a result, tie-breaking rules are not innocuous: they affect welfare even when stability is preserved. In contrast to the work of Erdil and Ergin, we allow for randomized matchings which lets us consider a significantly larger space of ex-ante outcomes. Our paper can be viewed as revisiting the agenda of Erdil and Ergin in the context of randomized matchings. Interestingly, our computational experiments show that realizing ex-ante efficiency improvements (i.e., \emph{before} random tie-breaking) can lead to substantially larger welfare gains in comparison to realizing efficiency improvements \emph{after} ties have already been broken (as proposed by Erdil \& Ergin).

Inspired by the observations of \citet{erdil2008s}, there is an increased focus on efficiency improvements while keeping stability consideration or allowing to weaken them. One example is the Efficiency-Adjusted DA mechanism (EADA) of \citet{kesten2010school}, and the improved algorithm by \citet{tang2014new}, which allows students to waive their priorities whenever this does not affect their own assignment, thus possibly resulting in unstable matchings. Other examples of mechanisms aiming to reduce inefficiencies in DA while limiting stability violations include the studies on legal matchings \citep{ehlers2020legal}, \citep{faenza2022legal}, weakly stable matchings \citep{tang2021weak}, and priority-efficient matchings \citep{reny2022efficient}. Most of these works are within the space of discrete matchings. 

A mechanism-design approach to stable \textit{random} matchings was initiated by~\citet{KestenUnver2015}, who studied the tension between stability and efficiency and proposed the Fractional Deferred Acceptance mechanism with Trading (FDAT) satisfying ex-ante stability which is a very demanding concept in the context of random matching. They find a constrained ordinally efficient solution within the class of ex-ante stable mechanisms.
In this paper, we focus on the more standard and widely-studied but weaker ex-post stability concept which also allows to access a larger set of matchings hence allowing for more efficiency.  Note that there is no clear theoretical efficiency comparison between ex-ante stability and the method by Erdil and Ergin, nor to DA with tie-breaking.







\citet{aziz2026ex} prove that when either side has ties in the preferences/priorities, testing ex-post stability is NP-complete.
The result even holds if both sides have dichotomous preferences.
They also consider stronger versions of ex-post stability (and prove that they can be tested in polynomial time).
Another related result is that checking whether a random assignment is ex-post Pareto optimal or not is coNP-complete \citep{AMXY15a}. If we focus on sd-efficiency, checking if a random assignment can be stochastically dominated by another random matching (or equivalently by an ex-post efficient random matching) can be done with an LP in polynomial time. On the other hand, checking if it can be dominated by an ex-post stable random matching is NP-complete which is our new contrasting result.

\citet{AzBr22c} provide a general vigilant eating approach to computing desirable matchings subject to general feasibility constraints. In our problem, capturing ex-post stability in the form of feasibility constraints is challenging as testing ex-post stability is coNP-complete~\citep{aziz2026ex}. 

As our main proposal is a new school choice mechanism with smart lotteries, our paper contributes to the literature on smart lotteries for matching problems. First, as already discussed in the introduction,  \cite{bronfman2018redesigning} study the Israeli resident allocation, and Pareto-improve upon RSD for estimated utilities by using matchings that assign couples close to the same or geographically close hospitals. Second, \cite{ashlagi2014improving} study how to increase neighborhood cohesion in school choice, which can, in turn, decrease busing costs. They find a lottery with the same marginals as RSD (i.e., no ex-ante improvement) by using matchings that maximize neighborhood cohesion. They prove the NP-hardness of this problem, and propose a Birkhoff-von Neumann-inspired heuristic, which they evaluate on school choice data from Boston. Third, \cite{DGHL23} focus on finding lotteries that maximize worst-case criteria. For example, they study the problem of finding a lottery with the same marginals as a given random matching, while minimizing the worst-case number of unassigned students by any of the matchings in the support of the lottery, and they show the NP-hardness of this problem. Their algorithmic framework is similar to ours, as they also propose a column generation in this setting, which they evaluate on Belgian school choice data. The paper is also related to the `best of both worlds' literature in which random outcomes satisfying desirable ex-ante and ex-post properties are sought~\citep{AFS+23a}.  

We conclude this section by providing some examples of applications other than school choice in which lotteries are used to break coarse priorities, and where our smart lottery mechanism could potentially be useful. In many national university admission schemes, the applicants' scores determine their priorities, as in Hungary, Spain, Chile, and Ireland, and lotteries might be used to break the ties for students with the same score, as in Ireland (see descriptions on the European applications at the website of the Matching in Practice network). Resident allocation schemes may also use common evaluation scores for the candidates, as in Scotland, or just pure lotteries, as in Israel, where the hospitals do not express their preferences over the candidates and a single lottery determines the picking order  \citep{bronfman2018redesigning}. In course allocation, coarse priorities can be based on certain student categories and then refined by lotteries, as at TU Munich \citep{BichlerMerting2021} or at ELTE university in Hungary \citep{Rusznak_etal2021}.  

\section{Preliminaries}

We consider many-to-one matching markets. Here, we have a bipartite market, where on one side, we are given a set of \emph{students} $N$, and on the other side, a set of \emph{schools} $S$. 
Each student $i\in N$ has a \emph{strict preference list} $\succ_i$.  We use the notation $s_j\succ_i s_k$ to denote that $i$ strictly prefers $s_j$ to $s_k$. We denote the set of acceptable schools of a student $i$ by $\acset (i)$. We assume that any acceptable school $s_j\in \acset (i)$ is strictly preferred by $i$ to $\emptyset$ (corresponding to not assigning $i$ to any school), but any unacceptable school is worse than $\emptyset$. 
 We define the \emph{rank} of school $s\in S$ for student $i$ (denoted by $\rank_i(s)$) to be the maximum possible size $k$ of a subset $\{ s_{j_1},\dots, s_{j_k}\}$, such that $s_{j_1}\succ_is_{j_2}\succ_i\cdots \succ_i s_{j_k}= s$. That is, the best school for $i$ has rank 1, the second best school rank 2, etc.
Each school $s_j\in S$ has a capacity $\capac[s_j]$ specifying the number of available positions at $s_j$. It also has a weak priority list $\succsim_{s_j}$ over the students $i\in N$ that find $s_j$ acceptable, where $k \sim_{s_j} \ell$ denotes that school $s_j$ is indifferent between students $k, \ell \in N$. For a school $s_j$, the set of students who find it acceptable is denoted by $\acset (s_j)$. 

A \emph{tie-breaking} of a priority list $\succsim_x$ is a strict priority list $\succ'_x$ that is consistent with $\succsim_x$, that is, $a\succ_x b$ implies $a\succ'_x b$.

\begin{definition} An instance $I$ of a \emph{many-to-one matching market} is represented by a tuple $(G,\capac,\succsim)$, where $G$ is a bipartite graph $G=(N,S,E)$, with the set of edges $E$ being defined by the mutual acceptability relations, that is, $(i,s_j)\in E$ if and only if $s_j\in \acset (i)$, $\capac$ is the capacity vector of the schools, while $\succsim$ is the preference profile consisting of the preference lists $\succ_i$ for $i\in N$ and the priority lists $\succsim_{s_j}$ for $s_j\in S$. 
\end{definition}
In this framework, we use the notation $E(v)$ to denote the set of edges incident to vertex $v$.

We say that a vector $M\in  \{ 0,1\}^{E}$ is a \emph{matching} of $G$, if it holds that for any $i\in N$, $\sum_{e\in E(i)}M(e)\le 1$ and for any $s_j\in S$, $\sum_{e\in E(s_j)}M(e) \le \capac[s_j]$. 
We say that a vector $p\in [0,1]^E$ is a \emph{random matching} of $G$, if $p$ can be written as a convex combination of matchings, that is, there exist matchings $M_1,\dots, M_k$, and nonnegative numbers $\lambda_1,\dots, \lambda_k$ such that $\sum_{l=1}^k\lambda_l = 1$ and $p=\sum_{l=1}^k\lambda_lM_l$. Using such a decomposition, one can also think of $p$ as a random lottery over the matchings $M_1,\dots, M_k$, where the probability of choosing the matching $M_l$ is $\lambda_l$.

Given a matching $M$, we use the notations $M(i)=\{ s_j\in S\mid M(i,s_j)=1\}$, that is the school that $i$ gets assigned to (or $\emptyset$ otherwise) and $M(s_j)=\{ i\in N\mid M(i,s_j)=1\}$, that is, the set of students assigned to $s_j$ in $M$. For a random matching $p$, we let $p(i)= \{ s_j\in S \mid p(i,s_j)>0\}$ and $p(s_j)=\{ i\in N\mid p(i,s_j)>0\}$.

\subsection{Stability and Pareto-optimality}

One of the most widely used fairness notions for matchings is weak stability. We say that a matching $M$ is \emph{weakly stable}, if there exists no student-school pair $(i,s_j)$, such that $s_j\succ_i M(i)$ and either (i) $|M(s_j)|<\capac[s_j]$ or (ii) $i\succ_{s_j} i'$ for some $i'\in M(s_j)$.

A random matching $p$ is called \emph{ex-post stable} if $p$ can be written as a convex combination of weakly stable matchings, that is, there exists weakly stable matchings $M_1,\dots, M_k$, and nonnegative numbers $\lambda_1,\dots, \lambda_k$ such that $\sum_{l=1}^k\lambda_l=1$ and $p=\sum_{l=1}^k\lambda_lM_l$. Note that this decomposition may not be unique. Furthermore, deciding if a random matching $p$ admits such a decomposition, i.e., it is ex-post stable, is NP-complete~\citep{aziz2026ex}.

 We say that a random matching $q$ \emph{stochastically Pareto-dominates} a random matching $p$ \emph{for the students} (or \emph{sd-dominates} for short), if for any student $i$ and school $s_j\in \acset (i)$, we have that $\sum_{s_l:s_l\succsim_i s_j}q(i,s_l)\ge \sum_{s_l:s_l\succsim_i s_j}p(i,s_l)$ and there exists at least one student $i$ and school $s_j\in \acset (i)$ such that strict inequality holds. That is, for any student $i$, and for any threshold school $s_j$, the probability that $i$ gets accepted to a school at least as good for her as $s_j$ cannot decrease. We say that a random matching $p$ is \emph{constrained-sd-efficient}, if it is (i) ex-post stable, and (ii) no other ex-post stable random matching $q$ sd-dominates $p$.

The main theoretical question of this paper is to settle the complexity of deciding whether a random matching $p$ is constrained-sd-efficient or not. Formally, we define the following decision problem.

\begin{problembox}[\myprob]
  \probleminput{A many-to-one market  $(G,\capac,\succsim)$ and an ex-post stable random matching $p$.}
  
  \problemquestion{Is $p$ constrained-sd-efficient?}
\end{problembox}


\subsection{Stable Improvement Cycles}

In the deterministic setting, the problem of deciding whether a weakly stable matching $M$ can be Pareto-improved (for the students) by another weakly stable matching $M'$ was studied by~\cite{erdil2008s}. They defined a matching $M$ to be \emph{constrained-efficient} if $M$ is weakly stable and no weakly stable matching $M'$ Pareto-dominates it for the students. 

Given a many-to-one matching market $(G,\capac,\succsim)$ and a stable matching $M$, \cite{erdil2008s} define a \emph{stable improvement cycle} as follows.

\begin{definition}
\label{def:SIC}
    Students $i_1,\dots, i_k$ form a \emph{stable improvement cycle} with respect to a weakly stable matching $M$, if
    \begin{enumerate}
    \item $M(i_j)\ne \emptyset$ for any $j\in [k]$,
        \item $M(i_{j+1})\succ_{i_j}M(i_j)$ for $j\in [k]$ with $i_{k+1}:=i_1$
        \item For any $j\in [k]$ and $s=M(i_{j+1})$, we have that $i_j\succsim_{s} i'$ for any $i'\in N$ with $s\succ_{i'}M(i')$. 
    \end{enumerate}
\end{definition}

\begin{theorem}[\cite{erdil2008s}]\label{thm:erdilergin}
    A weakly stable matching $M$ is constrained-efficient if and only if $M$ does not admit any stable improvement cycles. Furthermore, deciding whether $M$ admits a stable improvement cycle, and if yes, finding such a cycle can be done in polynomial time.
\end{theorem}

To find stable improvement cycles for a stable matching $M$, they define a directed envy graph $D_M$ as follows. 

\begin{itemize}
    \item The vertex set of $D_M$ is $N$, i.e., the set of students. 
    \item For any $i,j\in N$ we have a directed arc $(i,j)$ if and only if (i) $i,j$ are both matched in $M$, (ii) $M(j)\succ_i M(i)$ and (iii) $i\succsim_{M(j)} i'$ for any $i'\in N$ with $M(j)\succ_{i'}M(i')$.
\end{itemize}

Given a set $\set$ of disjoint stable improvement cycles in $D_M$, we can create a new stable matching $M/\set$ by eliminating the cycles in $\set$, i.e. by assigning $i_j$ to $M(i_{j+1})$ for each stable improvement cycle in $\set$~\citep{erdil2008s}.

\section{Theoretical Results}
\label{sec:theory}
In this section, we resolve the complexity of {\myprob} and related problems. The omitted proofs are in Appendix \ref{app:proofs}. 

\subsection{Efficiency Benefits of the {\name} Approach}

We start our study by first recalling the example that motivates our approach over existing approaches in terms of efficiency subject to ex-post stability. 
\begin{example}[Comparison with Previous Approaches]
\label{ex:EADA}
Consider Example~\ref{ex:basic}. In this example, our approach starts from a random matching $p$ and finds an sd-improvement $q=\frac{1}{2}(M_3+M_4)$ where $M_3$ and $M_4$ are specified in  Example~\ref{ex:basic}.
    {Note also that any mechanism that Pareto-improves the matchings ex-post, such as the stable improvement cycles by \citet{erdil2008s} and EADA by~\citet{kesten2010school}, cannot improve upon the random matching $p$ in this instance, as all tie-breakings lead to student-optimal stable matchings that are already Pareto-optimal for the students. This example illustrates that all these methods can be stochastically dominated by ex-post stable lotteries.}

To gain an intuition into why the sd-improvement is possible, observe that in DA with lottery all agents receive their third choice with positive probability. Take for example agent 1. If she fails to win the lottery against 2 for $s_1$ and student 3 also fails to get $s_2$, then they will have to compete for school $s_3$, and only one of them can get it, decided by a lottery. However, if we coordinate the first two lotteries for the top choices, by making them correlated, then we can avoid the competitions for the second choices. For instance, we can make sure that whenever student 1 wins the lottery for her top choice, then her potential competitor, student 3, loses the lottery for her top choice, and gets her second choice, and vice versa.\footnote{Note that the idea of correlated lotteries for school choice was also suggested in \citet{ashlagi2014improving}. However, their goal was to improve the probability of assigning the students from a neighborhood to the same schools, which can decrease the busing costs for the city.}
\end{example}

Next, we compare with the Fractional Deferred Acceptance with Trading (FDAT) mechanism of \citet{KestenUnver2015}. FDAT starts from a fractional version of the Deferred Acceptance outcome and then allows agents to trade probability shares of assignments. The initial fractional matching satisfies ex-ante stability, which is a stronger stability notion than ex-post stability, while the subsequent trading phase lets agents exchange portions of their assignments to improve efficiency according to their preferences, without violating ex-ante stability.

If we do allow randomization but target the stronger notion of ex-ante stability that is achieved by FDAT, then it severely restricts the welfare benefits in comparison to ex-post stable random matchings, as demonstrated in the following example. 

\begin{example}[Comparison with FDAT]
\label{sec:fdat}
We have eight students $1,2,3,4,5,6,7,8$ and eight schools $s_1,s_2,s_3,s_4,s_5,s_6,s_7,s_8$ with unit capacities, and with the following preferences and priorities.

    \begin{center}
        \begin{minipage}{0.3\textwidth}
\centering
\begin{tabular}{rl}
      $1:$ & $s_1\succ s_3 \succ s_4 \succ s_2$ \\
      $2:$ & $s_1\succ s_4 \succ s_3\succ s_2$ \\
      $3:$ & $s_2\succ s_4\succ s_3 \succ s_1$  \\
      $4:$ & $s_2\succ s_3\succ s_4 \succ s_1$ \\
      $5:$ & $s_5\succ s_4\succ s_6$ \\
      $6:$ & $s_7\succ s_3\succ s_8$ \\
      $7:$ & $s_5\succ s_6$  \\
      $8:$ & $s_7\succ s_8$ 
      \end{tabular}
\end{minipage}
\hspace{0.03\textwidth}
\begin{minipage}{0.3\textwidth}
\centering
\begin{tabular}{rl}
      $s_1:$ & $[3,4]\succ [1,2]$ \\
      $s_2:$ & $[1,2]\succ [3,4]$ \\
      $s_3:$ & $3\succ 2\succ [1,6]\succ 4$ \\
      $s_4:$ & $4\succ 1 \succ [2,5]\succ 3$\\
      $s_5:$ & $[7,5]$\\
      $s_6:$ & $7\succ 5$\\
      $s_7:$ &$[6,8] $\\
      $s_8:$ & $8\succ 6$.
      \end{tabular}
\end{minipage}
\end{center}
Here, the notation $[i,j]$ for some agents $i,j$ means that they have the same priority.

FDAT outputs the random matching $p$ where students $7$ and $8$ receive their first two choices with probability $\nicefrac{1}{2}$, while all other students receive their first and third choices with probability $\nicefrac{1}{2}$. In particular, $p(1,s_1)=p(1,s_4)=p(2,s_1)=p(2,s_3)=p(3,s_2)=p(3,s_3)=p(4,s_2)=p(4,s_4)=p(5,s_5)=p(5,s_6)=p(6,s_7)=p(6,s_8)=p(7,s_5)=p(7,s_6)=p(8,s_7)=p(8,s_8)=\nicefrac{1}{2}$ and $p=0$ otherwise. 

Then, FDAT finds a possible improvement, where $1$ and $2$ trade their third choices for their second choices, resulting in a random matching where students $1$ and $2$ receive their first two choices with probability $\nicefrac{1}{2}$. 

Now, the key property of this instance is that students $3$ and $4$ could also switch their third-ranked schools with each other to obtain their second-ranked schools. However, this cannot happen in the FDAT, as the resulting random matching violates ex-ante stability, because the pairs $(5,s_4)$ and $(6,s_3)$ would block in the ex-ante sense. 

We will show that the random matching $p'$ that is obtained by allowing students $3$ and $4$ to swap probabilities for their third choices is ex-post stable. Note that $p'$ assigns students $5$ and $6$ to their first and third choices with probability $\nicefrac{1}{2}$, while all other students receive their first two choices with probability $\nicefrac{1}{2}$. Hence, $p'$ clearly sd-dominates FDAT as students $3$ and $4$ are strictly better off, while the other students are indifferent between both random matchings.

Observe that $p'$ is ex-post stable, as $p'= \frac{1}{2}(M_1+M_2)$, where $M_1$ and $M_2$ are the following weakly stable matchings:
\begin{itemize}
    \item $M_1= \{ (1,s_1),(2,s_4),(3,s_2),(4,s_3),(5,s_6),(6,s_7),(7,s_5),(8,s_8)\}$,
    \item $M_2 = \{ (1,s_3),(2,s_1), (3,s_4),(4,s_2), (5,s_5),(6,s_8),(7,s_6),(8,s_7)\}$.
\end{itemize}

We show that $p'$ weakly sd-dominates the result of the DA with uniform tie-breakings. 
For this, first observe that every student has at most $\frac{1}{2}$ probability of being matched to their first choice, since each top school is a top school of two students, where these two are tied. This implies that for students $1,2,3,4,7,8$, $p'$ sd-dominates the uniform tie-breaking DA output. Next, suppose that $5$ is with $s_4$ in the DA output after some tie-breaking. Then, students $1$ and $4$ must be with better schools than $s_4$. If both of them are at their top schools ($s_1$ and $s_2$), then either student $2$ is unmatched and blocks with $s_2$ or student $3$ is unmatched and blocks with $s_1$. Otherwise, student $1$ or $4$ must be at $s_3$. Then, as one of students $2$ and $3$ must be unmatched, they block with $s_3$, contradiction. Hence, $5$ is never with $s_4$ in the DA output after some tie-breaking. Similarly, by a symmetric argument, $6$ is never with $s_3$ in the DA output after some tie-breaking. Hence, for them, the uniform tie-breaking DA output is also sd-dominated by $p'$.

Also, $p'$ cannot be sd-improved, because if $5$ or $6$ increase their probability for a top-2 school, then a student there must get a worse than second school with positive probability.
Hence, $p'$ is a possible optimal output of our approach.
\end{example}

\subsection{Constrained-sd-efficiency}

We start by examining the complexity of \myprob. First, let us prove a  useful lemma, which will help us restrict the support set of constrained-sd-efficient matchings to constrained-efficient matchings.

\begin{lemma}\label{thm:onlystuopts}
   Take a many-to-one matching instance $I=(G,\capac, \succsim)$. If $p$ is constrained-sd-efficient, then in any decomposition $p=\sum_l \lambda_lM_l$ of $p$ into weakly stable matchings, each $M_i$ is constrained-efficient and a student optimal stable matching with respect to some tie-breaking.
\end{lemma}
\begin{proof}

Suppose for the contrary that $p=\sum_i \lambda_iM_i$ but $M_j$ is not a student optimal stable matching with respect to any tie-breaking of the preference profile $\succsim$. Take a tie-breaking $\succ'$, satisfying that for each school $s$, if $i\in M(s)$ but $i'\notin M(s)$ and $i'\sim_s i$, then $i\succ'_s i'$. 

It is straightforward to verify that $M_j$ is a stable matching with respect to the profile $\succ'$.
By our assumption, we know that the student optimal stable matching $M_j'$ with respect to $\succ'$ is weakly better for all students and strictly better for at least one. By the construction of the tie-breaking $\succ'$, $M_j'$ gives a weakly stable matching that is weakly better for every student and strictly better for at least one. Then, $q=\sum_{l\ne j}\lambda_lM_l + \lambda_jM_j'$, is an ex-post stable random matching  (as witnessed by this decomposition) that SD-dominates $p$, a contradiction to $p$ being constrained-sd-efficient.

Similarly, if $M_j$ is not constrained-efficient, then there exists a weakly stable matching $M_j'$ that Pareto-dominates $M_j$ for the students, so $q=\sum_{l\ne j}\lambda_lM_l + \lambda_jM_j'$, is an ex-post stable random matching  that SD-dominates $p$, contradiction.
\end{proof}

We turn to our main theoretical results, concerning the computational complexity of \myprob. First, we start with a simple observation for the case where all priority lists are strict. 

\begin{observation}
    \myprob\ is solvable in polynomial time, if all priority lists are strict. 
\end{observation}

Indeed, in this case, there exists a unique student-optimal stable matching $M$~\citep{gale1962college}, where all students receive the best possible school in any stable matching simultaneously. Hence, $M$ SD-dominates any ex-post stable random matching $p$, thus $M$ is also the unique constrained-sd-efficient random matching.




Sadly, even the introduction of ties on one side of the market makes the problem hard, as we now show.

\begin{theorem}\label{thm:improve-NPh}
    \myprob\ is coNP-complete, even if the students have preferences of length at most 3 and every capacity is 1. 
\end{theorem}

This result is in strong contrast with the polynomial result in Theorem~\ref{thm:erdilergin} for the deterministic case. At the same time, it is also in contrast with the polynomial complexity of verifying whether an arbitrary random matching is sd-efficient, which can be verified by checking for probabilistic trading cycles \citep{BoMo01a}.



In practice, the initial ex-post stable matching $p$ is usually obtained by uniformly randomly sampling some tie-breakings, computing the DA outcomes, and weighing them equally. This is to mitigate the computational bottleneck of computing the average of all possible (exponentially many) DA outcomes for all possible tie-breakings, which we shortly show to be NP-hard in Theorem~\ref{thm:uniform}. Hence, next we consider whether the coNP-completeness of \myprob\ remains even if $p$ is restricted to such matchings. 

\begin{theorem}
    \myprob\ is coNP-complete, even if the students have preferences of length at most 3 and $p$ is a uniform combination of DA outputs with different tie-breakings. 
\end{theorem}
\begin{proof}
   Consider the constructed ex-post stable random matching $p$, where we have shown that $p=\frac{1}{3}(M_1+M_2+M_3)$ for some known weakly stable matchings. If one of them, say $M_1$ is not a student optimal stable matching with respect to some tie-breaking, then by Lemma~\ref{thm:onlystuopts} and Theorem~\ref{thm:erdilergin}, we can find in polynomial time a weakly stable matching $M_1'$ that Pareto-dominates $M_1$ for the students and conclude that $q=\frac{1}{3}(M_1'+M_2+M_3)$ SD-dominates $p$. Hence, the restriction of the problem to inputs where all of $M_1,M_2,M_3$ is a student optimal stable matching with respect to some tie-breaking is still coNP-complete.
\end{proof}

\begin{theorem}\label{thm:uniform}
  It is NP-hard to compute the uniform combination of all DA output weakly stable matchings with respect to some tie-breaking.
\end{theorem}
\begin{proof}
We reduce from \comsmti, shown to be NP-complete by~\citet{Manlove_etal2002}.

\begin{problembox}[\comsmti]
  \probleminput{A stable marriage instance $I=(U,W;E ;(\succ_u)_{u\in U},(\succsim_w)_{w\in W})$, where each $u\in U$ has a strict preference list.}
  
  \problemquestion{Is there a complete stable matching in $I$?}
\end{problembox}

Let $I=(U,W;E ;(\succ_u)_{u\in U},(\succsim_w)_{w\in W})$ be an instance of \comsmti. We create many-to-one matching market $I'=(G,\capac, \succsim)$, by copying the graph and the preferences from $I$ (identifying the students $N=\{ 1,\dots, n\}$ with $U=\{ u_1,\dots, u_n\}$ and the schools $S=\{ s_1,\dots,s_m\}$ with $W=\{ w_1,\dots, w_n\}$) and setting $\capac \equiv 1$. Students have strict preferences, as $\succ_u$ is strict for all $u\in U$.
Additionally, we add a dummy school $t$ that is the worst acceptable for all students and ranks every student in a tie. Hence, we have $N=\{ u_1,\dots, u_n\}$ and $S=\{ w_1,\dots, w_n\} \cup \{ t\}$. 

If there exists a complete weakly stable matching in $I$, then let $M$ be such a matching. Then, if we break the ties in a way such that for each school $s\in S$, if $i\in M(s)$ but $i'\notin M(s)$ and $i'\sim_s i$, then $i\succ'_s i'$, then $M$ is easily seen to be stable. Hence, with respect to this tie-breaking, the DA-output student optimal stable matching must match all students to strictly better schools than $t$.

Hence, the sum of values on the incident edges to the dummy school is less than 1 in the uniform combination of DA-outputs.

In the other direction, if the sum of the values on incident edges to the dummy school $t$ is less than 1 in the uniform combination of DA outputs, then there must be a tie-breaking, where the student optimal stable matching $M$ does not match anyone to $t$. Since it is acceptable to every student, each must be matched to a school corresponding to some $w_j\in W$ in $I$. It is easy to see that $M$ gives a complete weakly stable matching in $I$.
\end{proof}

\subsection{Related Questions for Deterministic Matchings}
\label{subsec:questions_deterministic}

Motivated by the positive result of Erdil and Ergin for the deterministic case~\citep{erdil2008s} (Theorem~\ref{thm:erdilergin}), we also study a related question, where instead of Pareto-dominating a weakly stable matching $M$ with a different weakly stable matching, we ask whether an arbitrary matching $M$ can be Pareto-dominated by a weakly stable matching or an ex-post stable random matching. Sadly, these questions are also computationally challenging.

\begin{theorem}
\label{th:EE_stronger}
    Given a matching $M$, it is NP-complete to decide whether there is a weakly stable matching $N$ that Pareto-dominates $M$ for the students. It is also NP-complete to decide whether there is an ex-post stable random matching $p$ that Pareto-dominates $M$ for the students. 
\end{theorem}
\begin{proof}
We reduce from \comsmti. 
Let $I=(U,W;E ;(\succ_u)_{u\in U},(\succsim_w)_{w\in W})$ be an instance of \comsmti. We create many-to-one matching market $I'=(G,\capac, \succsim)$, by copying the graph and the preferences from $I$ (identifying the students with $U$ and the schools with $W=\{ w_1,\dots, w_n\}$) and setting $\capac \equiv 1$. Students have strict preferences, as $\succ_u$ is strict for all $u\in U$.
Additionally, we add a dummy school $t_i$ for all $u_i\in U$, that is the worst acceptable for student $u_i$ and ranks every student in a tie, except that we create a corresponding dummy student $y_i$ that considers only $t_i$ acceptable and is best for $t_i$. Hence, we have $N=\{ u_1,\dots, u_n\} \cup \{ y_1,\dots, y_n\}$ and $S=\{ w_1,\dots, w_n\} \cup \{ t_1,\dots, t_n\}$. 


The initial matching $M$ is $\{ (u_i,t_i)\mid i\in [n]\}$.

In any weakly stable matching, $y_i$ and $t_i$ are matched for each $i\in [n]$ as they are mutual top choices. Hence, to Pareto-improve the students, each $u_i$ must be matched and necessarily to a school $w_i$. Such a matching between $U$ and $W$ must be a complete weakly stable matching in $I$. For the case, where we want to dominate with an ex-post stable random matching $p$, we still must satisfy that $u_i$ is matched with probability 1, hence all weakly stable matchings in the support of $p$ must match all $u_i$ students to schools from $W$ (by stability), so all of them must be complete weakly stable matchings in $I$.

Conversely, given a complete weakly stable matching $M'$ in $I$, we can extend $M'$ with $\{ (y_i,t_i)\mid i\in [n]\}$ and obtain a weakly stable matching (and thus also an ex-post stable matching), which Pareto-improves the students upon $M$.
\end{proof}

We also study how to select one of the (possibly many) stable improvement cycles in each iteration of the algorithm by~\cite{erdil2008s}. Erdil and Ergin did not consider the problem of finding ``optimal'' stable improvement cycles, i.e., Pareto-improving the weakly stable matching $M$ in such a way that the average rank of the students is minimized. In this paper, we extend their work by analyzing such problems computationally.

The first question that arises is whether we can find an ``optimal'' set of stable improvement cycles in~$D_M$?

\begin{problembox}[\bestgreedyEE]
  \probleminput{A many-to-one matching market  $(G,\capac,\succsim)$, a weakly stable matching $M$ and a number $R$.}
  
  \problemquestion{Is there a set $\set$ of disjoint stable improvement cycles in $D_M$, such that for $M'=M/\set$ $\avgrank(M)-\avgrank(M')\ge R$?}
\end{problembox}

\bestgreedyEE\ can be solved in polynomial time, as it straightforwardly reduces to a minimum cost circulation problem with unit edge and vertex capacities, using that such a problem always admits an integer optimal solution~\citep{Ahuja1993,hoffman2009integral}. The weight of an arc $(i,s_j)$ should be set to the amount in the average rank the switch of $i$ from $M(i)$ to $s_j$ leads to (i.e., $\frac{\rank_i (s_j)-\rank_i (M(i))}{n}$).

The more interesting problem is to find a stable improvement that is globally optimal, i.e. a stable matching that Pareto-improves the students from $M$ and leads to the best average rank. The difficulty here is that even though we eliminate a set of disjoint improvement cycles that is optimal in the current envy graph $D_M$, there is no guarantee that such an improvement does not destroy all paths of eliminating improvement cycles that lead to a global optimum. In fact, we will show that this becomes NP-hard.

\begin{problembox}[\bestimproveEE]
  \probleminput{A many-to-one matching market $(G,\capac,\succsim)$, a weakly stable matching $M$ and a number $R$.}
  
  \problemquestion{Is there a stable matching $M'$, such that $M'$ is a Pareto-improvement for the students and $\avgrank(M)-\avgrank(M')\ge R$?}
\end{problembox}

\begin{example}
We illustrate with an example that greedily picking the best stable improvement cycle, or the best disjoint set of such cycles, may not be optimal.
    We have six students $1,2,3,4,5,6$ and six schools $s_1,s_2,s_3,s_4,s_5,s_6$ with the following preferences and priorities. 

    \begin{center}
        \begin{minipage}{0.25\textwidth}
\centering
\begin{tabular}{rl}
     $1:$ & $s_2\succ s_4$  \\
      $2:$ & $s_1\succ s_2\succ s_3\succ s_5$   \\
      $3:$ & $s_1\succ s_3\succ s_6$ \\
      $4:$ & $s_4\succ s_1$  \\
      $5:$ & $s_5\succ s_2$ \\
      $6:$ & $s_6\succ s_3$ 
    \end{tabular}
\end{minipage}
\hspace{0.01\textwidth}
\begin{minipage}{0.25\textwidth}
\centering
\begin{tabular}{rl}
     $s_1:$ & $4\succ [2,3]$  \\
      $s_2:$ & $5\succ 1\succ 2$\\
      $s_3:$ & $6\succ 2\succ 3$ \\
      $s_4:$ & $1\succ 4$ \\
      $s_5:$ & $2\succ 5$ \\
      $s_6:$ & $3\succ 6$
    \end{tabular}
\end{minipage}
\end{center}
Take the matching $M=\{ (1,s_4),(2,s_5),(3,s_6),(4,s_1),(5,s_2),(6,s_3)\}$. It is weakly stable, as all schools have a top student. 

$D_M$ has arcs $\{ (1,5),(2,6),(2,4), (3,4),(4,1),(5,2),(6,3)\}$. We have two possibilities: stable improvement cycle 
\[
C_1=\{ (1,5),(5,2),(2,6),(6,3),(3,4),(4,1)\}
\]
or
\[
C_2= \{ (1,5),(5,2),(2,4),(4
,1)\}.
\]
It is clear that $C_1$ is a better choice greedily, as it improves the average rank more. However, that leads to constrained-efficient matching $M_1=\{ (1,s_2),(2,s_3),(3,s_1),(4,s_4),(5,s_5),(6,s_6)\}$. 

If we choose cycle $C_2$ instead, then we reach $M_2=\{ (1,s_2),(2,s_1),(3,s_6),(4,s_4),(5,s_5),(6,s_3)\}$. Then, we have that $D_{M_2}$ still has arcs $\{ (3,6),(6,3)\}$ which form a stable improvement cycle $C_3$. By eliminating this cycle as well, we get $M_3=\{ (1,s_2),(2,s_1),(3,s_3),(4,s_4),(5,s_5),(6,s_6)\}$. Then, $\frac{7}{6}=\avgrank (M_3)<\avgrank (M_1) = \frac{8}{6}$.

\end{example}


We show that checking whether there exists a weakly stable matching $M'$ that improves the average rank compared to an initial weakly stable matching $M$ by at least $R$ is NP-complete.

\begin{theorem}\label{th:bestimproveEE}
    \bestimproveEE\ is NP-complete.
\end{theorem}

\subsection{Strategy-Proofness and Constrained Efficiency}
A random mechanism is \textit{strategy-proof} if no student can manipulate by misreporting, for any cardinal preferences that are consistent with the students' preferences \citep{BoMo01a}. \cite{erdil2008s} illustrate on an example that no strategy-proof, constrained-efficient (in the integral sense) and stable mechanism exists. Their example satisfies that there exist two constrained-efficient stable matchings, however, students $1$ and $2$ can both submit a false preference list, which forces that the only constrained-efficient stable matching will be the one that is preferred by them.

For completeness, we describe their example here. 

There are three students, $1,2,3$ and three schools $s_1,s_2,s_3$. The preferences of the students are $\succ_1=\succ_2= s_2\succ s_3\succ s_1$, and $\succ_3 = s_1\succ s_2\succ s_3$. The priorities of the schools are $\succ_{s_1}=1\succ 2\succ 3$, $\succsim_{s_2}=3\succ [1,2]$ and $\succ_{s_3}=3\succ 2\succ 1$. 

The two possible student optimal (and constrained-efficient) stable matchings are $M_1=\{ (1,s_2),(2,s_3),(3,s_1)\}$ and $M_2=\{ (1,s_3),(2,s_2), (3,s_1)\}$. If either student 1 or 2 submits $s_2\succ s_1\succ s_3$, then the only constrained-efficient stable matching will be the one among $M_1,M_2$ that is better for the strategizing student.

Mirroring this, we get the following result using Lemma~\ref{thm:onlystuopts}.

\begin{theorem}
There is no strategy-proof, ex-post stable, and constrained-sd-efficient random mechanism.
\end{theorem}
\begin{proof}
    By Lemma~\ref{thm:onlystuopts}, we get that such a mechanism must output a convex combination of the two constrained-efficient stable matchings. However, as shown by~\citet{erdil2008s}, student $1$ or $2$ (one of them prefers one of the constrained-efficient stable matching to the other constrained-efficient stable matching in the support of the output) can submit a false preference list to ensure that the only constrained-efficient stable matching will be the one he prefers, which by Lemma~\ref{thm:onlystuopts} will also be the unique constrained-sd-efficient matching. This contradicts strategy-proofness.
\end{proof}

\section{A Column Generation Framework}
\label{sec:column_gen}
In this section, we introduce a framework to solve the optimization variant of \myprob, namely to find the ex-post stable random matching of minimal average rank that stochastically dominates a given random matching $p$. A first possible method is to formulate an integer program with separate decision variables for each of the weakly stable matchings in the support of a random matching $q$ minimizing the average expected rank among all random matchings that sd-dominate $p$. This formulation is described in Appendix \ref{sec:IP}. However, this formulation is only capable to solve very small instances in practice; it aims to simultaneously find $n^2+1$ minimum-weight weakly stable matchings, which are NP-hard to generate on their own \citep{Manlove_etal2002}, and finding them simultaneously creates symmetry and causes difficulties in providing useful bounds. 

In response, we propose a \emph{column generation} framework. The idea behind column generation is to first generate a subset of the weakly stable matchings, and to then find a convex combination of the matchings in this subset that has minimal average rank while sd-dominating $p$. Next, in a separate problem, called the \emph{pricing problem}, we verify whether the found solution is optimal over all random matchings (i.e., including random matching with alternative weakly stable matchings in their support). If not, additional weakly stable matchings are generated and the procedure is repeated.

In this section, we explicitly include a dummy school $\varnothing$ in~$S$ that represents remaining unassigned. This dummy school has capacity $c[\varnothing] = |N|$ and arbitrary priorities, and is acceptable to every student, but is always strictly worse than any other acceptable school. Hence, we can assume that any weakly stable matching is student-perfect. We adopt the convention that $\text{rank}_i(\varnothing) = |S|+1$, for all students $i\in N$. 

Denote the set of all weakly stable matchings by $\M$. Let $\tilde{\M}\subseteq \M$ denote a subset of the weakly stable matchings. These could be, for example, the weakly stable matchings that were generated in order to estimate the probabilities of DA with tie-breaking. 

The following linear program $[P(\tM)]$ finds a heuristic solution for the ex-post stable random assignment with minimal rank that stochastically dominates a given random matching $p$, while only allowing for weakly stable matchings in $\tM$ in the decomposition. The decision variables $\lambda_\ell$ denote the weight of matching $M_\ell\in\tM$ in the final decomposition.

\allowdisplaybreaks

\begin{align}
&[\text{P($\tM$)}]&& \text{min} & \sum_{\ell:M_\ell\in\tilde{\M}}\lambda_\ell\left(\sum_{(i,s_k)\in E}  M_l(i,s_k)\cdot \rank_i(s_k)\right)  & &\label{IP:CG1} \\
&&&\text{s.t.}& \sum_{\ell:M_\ell\in\tilde{\M}}\sum\limits_{s_{k'}: s_{k'}\succsim_i s_k} \lambda_\ell\cdot M_\ell(i,s_{k'}) & \ge \sum\limits_{s_{k'}: s_{k'}\succsim_i s_k}p(i,s_{k'}) & ((i,s_k)\in E)\label{IP:CG2}\\
  &&  & & \sum_{\ell:M_\ell\in\tilde{\M}}\lambda_\ell & = 1 & \label{IP:CG3}\\
&&&&\lambda_\ell & \geq 0  & (l: M_l\in\tM)
\label{IP:CG4}
\end{align}

Constraints (\ref{IP:CG2}) enforce that the found random matching stochastically dominates $p$, while Constraints (\ref{IP:CG3}) and (\ref{IP:CG4}) enforce that a feasible decomposition is obtained. Denote by $\mu_{ik}\geq 0$, and $\delta$ the dual variables of Constraints (\ref{IP:CG2}), and (\ref{IP:CG3}), with $(i,s_k)\in E$. The dual of formulation [P($\tM$)] is represented by formulation [D($\tM$)].
\begin{align}
    &[\text{D($\tM$)}]&&\text{max} & \sum_{(i,s_k)\in E} \mu_{ik} \left(\sum\limits_{s_{k'}: s_{k'}\succsim_i s_k} p(i,s_{k'})\right) &+ \delta &\label{con:D1}\\
    &&&\text{s.t.} &\sum_{(i,s_k)\in E}\mu_{ik} \left(\sum\limits_{s_{k'}: s_{k'}\succsim_i s_k} M_\ell(i,s_{k'}) \right) + \delta &\leq \sum_{(i,s_k)\in E} M_\ell(i,s_k)\cdot \rank_i(s_k) & (\ell: M_\ell\in\tM)\label{con:D2}\\
    &&&&\mu_{ik}&\geq 0 &((i,s_k)\in E)\label{con:D3}
\end{align}

Denote by [D] the dual formulation that contains one constraint of type (\ref{con:D2}) for each weakly stable matching in $\M$. We know that the optimal solution of [P($\tM$)] over the subset $\tM$ is optimal over the set $\M$ of \emph{all} weakly stable matchings if its optimal objective value is equal to the optimal objective value of [D]. To verify optimality of a feasible solution of [P($\tM$)], with dual variables $\mu_{ik}$, and $\delta$, we can check whether there exists a weakly stable matching $M\in\M\setminus\tM$ that violates Constraint (\ref{con:D2}) using the following formulation, which we call the \emph{pricing problem}.
\begin{align}
    &\text{max} &\sum_{(i,s_k)\in E}&\left(\mu_{ik} \left(\sum\limits_{s_{k'}: s_{k'}\succsim_i s_k} M(i,s_{k'}) \right)-M(i,s_k)\cdot \rank_i(s_k)\right) + \delta \label{con:pricing_obj} \\
    &\text{s.t.}&M&\in\M \label{con:pricing_con}
\end{align}

If this formulation finds a matching $M^*$ with an objective value strictly larger than zero, then the found solution of [P$(\tM)$] was not optimal over all weakly stable matchings in $\M$. In response, the found matching $M^*$ is added to the subset $\tM$, and the formulation [P($\tM$)] is solved again. This iterative approach continues until the formulation [P($\tM$)] finds a solution for which the optimal objective value of the pricing problem is less than or equal to zero, in which case we have a guarantee that the solution we found in [P($\tM$)] over the subset $\tM$ of weakly stable matchings was optimal over the set $\M$ of all weakly stable matchings.

Note that any upper bound on the objective value of the pricing problem can be used to construct a lower bound on the optimal objective value $z^*$ of $[P(\mathcal{M})]$, which uses all weakly stable matchings. Let $z(\tilde{\M})$ denote the optimal objective value of $[P(\tilde{\mathcal{M}})]$. Let $v(\tilde{\M})$ denote the optimal objective value of the pricing problem corresponding to $[P(\tilde{\mathcal{M}})]$. Then, the convexity constraint (\ref{IP:CG3}) implies that one cannot improve the objective value of $z(\tilde{\M})$ by more than the smallest reduced cost of any weakly stable matching, which equals $-v(\tilde{\M})$ \citep{lubbecke2005selected}, i.e., 
\begin{equation}
    z^* \geq z(\tilde{\M}) - v(\tilde{\M}).
\end{equation}

When the pricing problem is not solved to optimality, $v(\tilde{\M})$ can be replaced by any valid upper bound on the pricing objective value, yielding a weaker bound on $z^*$.

To enforce that the matching found by the pricing problem is weakly stable, i.e., belongs to $\M$, several different formulations can be used (see \cite{delorme2019mathematical} for an overview). We consider constraints imposing so-called \emph{cut-off ranks}, given their overall strong performance (see also \cite{agoston2016integer}). 

Let $\overline{r}_j$ denote the number of indifference classes in the priorities of school $s_j$, and let $r_{ij}\in \{1,\dots,\overline{r}_j\}$ denote the rank of student $i$ in the priorities of school $s_j$. Note that the lower the rank of a student at a school, the higher her priority is there, so $i'\succsim_{s_j}i$ if and only if $r_{i'j}\leq r_{ij}$.
Given a matching $M$, the \textit{cut-off rank} of a school $s_j$, denoted by $y_j^M$, is the rank of student with the highest rank at school $s_j$ among those students who are assigned to $s_j$ at $M$, i.e., $y_j^M = \max\{r_{ij}: M(i,s_j) = 1\}$.

The following sets of constraints will impose that only students who have a rank that is weakly lower than the cutoff rank at a certain school can be admitted there. We link the cutoff ranks with the induced matching using the following constraints.
\begin{align}
    && y_j^M &\geq M(i, s_j)\cdot r_{ij}  &((i,s_j)\in E)\label{con:cutoff_begin}\\
    && r_{ij} &\geq y_j^M - \left(\sum_{s_{j'}\succsim_i s_j}M(i,s_{j'})\right) \cdot(\overline{r}_j+1) &((i,s_j)\in E)
    \end{align}
These constraints already provide the weak fairness of the induced matching, that is no student can be admitted to a school if another student with higher priority (and thus with lower rank) is rejected. To achieve weak stability, we also have to ensure non-wastefulness, which means that no student can be rejected from a school if a seat is left empty there. We can achieve this by enforcing the cutoff rank of a school to be equal to $\overline{r}_j + 1$ at that school if the capacity of this school is not fully filled by using the following two sets of constraints, where $f_j\in\{0,1\}$ is an auxiliary binary variable.
\begin{align}
    && f_j\cdot c[s_j] &\leq \sum_{(i,s_j)\in E}M(i,s_{j})  \;\leq c[s_j]  &(s_j\in S)\\
    && y_{j}^M &\geq (1-f_j)\cdot(\overline{r}_j+1) &(s_j\in S)\\
    && \sum_{(i,s_j) \in E} M(i,s_j) & = 1 & (i\in N)\\
    &&M(i, s_j) &\in \{0,1\} &((i,s_j) \in E)\\
    &&f_j&\in\{0,1\} &(s_j\in S) \label{con:cutoff_end}
    \end{align}

In Appendix \ref{app:equal_treatment}, we discuss how this framework can be extended to incorporate \emph{equal treatment of equals}, i.e., impose that students with identical preferences and priorities receive the same assignment probabilities.

\section{Computational Experiments}
\label{sec:comp_exper}




In this section, we evaluate how much can be gained in terms of efficiency by applying the proposed methods both on generated and real-world data.

\subsection{Evaluated Methods}
\label{subse:methods}
We evaluate the following methods. First, let \textbf{EE} refer to the random matching that is found by iteratively resolving stable improvement cycles (SICs, Definition~\ref{def:SIC}), starting from the matching found by DA after random tie-breaking, as proposed by \citet{erdil2008s}.

Second, let \textbf{$X$-\name-heur} refer to the random matching that sd-improves upon a random matching $X$ by only using a set $\tilde{\mathcal{M}}$ of sampled matchings, as described in Section~\ref{sec:column_gen}. We will evaluate three different random matchings which we will sd-improve upon: DA with random tie-breaking (DA), DA with random tie-breaking and resolving the SICs (EE), and the Efficiency Adjusted DA mechanism (EADA) proposed by \cite{kesten2010school}. In our experiments, the set $\tilde{\mathcal{M}}$ will contain the matchings that were sampled by applying DA to different tie-breaking rules, and the matchings obtained by resolving stable improvement cycles in those matchings. The solution of this method is found by solving formulation $[\text{P($\tM$)}]$ (i.e., the first step of the column generation procedure) in Section~\ref{sec:column_gen}, and is therefore a heuristic.

Similarly, let \textbf{$X$-\name-CG} denote the random matching that is found by running the column generation framework from the initial subset $\tilde{\mathcal{M}}$ of sampled matchings while sd-dominating a random matching $X$.

\subsection{Implementation Details}



\label{subsec:implementation}
All experiments are implemented in Python 3.13.5 and run on an Intel Xeon W-2123 processor running at 3.60 GHz, with 32GB of RAM on a Windows 11 64-bit OS. All linear and integer linear programs are solved using Gurobi 12.0. The assignment probabilities of DA with random tie-breaking are estimated by sampling $1,000$ random orderings of the students, and applying DA with single tie-breaking to the resulting instances.\footnote{Since Random Serial Dictatorship for house allocation can be viewed as a special case of DA with single tie-breaking, it follows from \citep{ABB13b} that computing the ex-ante random matching is \#P--complete.}
Next, the SICs are resolved for those generated matchings, thus obtaining the initial subset of weakly stable matchings. If [P($\tM$)] cannot find a feasible solution using the matchings in $\tilde{\mathcal{M}}$ to sd-dominate a random matching $X$, we add an artificial column $M^0$ with a very high objective coefficient, such that $M^0_{ij}=1$ for all $(i,s_j)\in E$. Whenever $\lambda_0=0$, column $M^0$ is removed, and a feasible solution is found. In each iteration of the pricing problem, 40 weakly stable matchings were generated. A time limit of 10 minutes was imposed on the column generation for each instance of the generated data. For generated instances, each data point is the average over 10 randomly generated instances. The implementation of SICs was inspired by the published code by \citet{erdil2019replication}.\footnote{We identified a bug in the code by \citet{erdil2019replication}, which caused it to sometimes return matchings that were not weakly stable. \citet{demeulemeester2026comment} discusses this in more detail, and proposes a modified implementation.} The implementation of EADA is inspired by the published code by \citet{faenza2022legal}. All ranks are reported under the convention that being unassigned yields rank $|S|+1$. To determine sd-dominance, we adopt a minimum improvement threshold of $10^{-5}$.

\subsection{Data Generation}
\label{subsec:data_gen}
We follow the same data generation procedure as \citet{erdil2008s}, implemented through their code, which is available online \citep{erdil2019replication}. In short, let $\ell^i, \ell^j \in \mathbb{R}^2$ denote the locations of the students $i\in N$ and the schools $s_j \in S$, which are generated uniformly at random on $[0,1]\times[0,1]$. We refer to a dummy student with average tastes by $i=0$. The utility of student $i\in N$ for school $s_j\in S$ is determined by 
$$U_{ij} = -\beta d(\ell^i, \ell^j) + (1-\beta) \left(\alpha Z_{0,j} + (1-\alpha) Z_{ij} \right),$$
where $d(x,y)$ denotes the Euclidean distance between points $x,y\in\mathbb{R}^2$, $Z_{ij}$ are i.i.d.\ normally distributed random variables with mean zero and variance one, and $\alpha,\beta\in[0,1]$ are input parameters. The parameter $\alpha$ captures the correlation in the students' preferences, and the parameter $\beta$ captures how sensitive the students' preferences are to locational proximity. Each student is assumed to be in the walk zone of the school closest to them. The priorities of the schools only consist of two indifference classes determined by their walk zone.

\subsection{Comparison with Erdil \& Ergin (2008)}
\label{subsec:comp_EE}

\begin{figure}[!tb]
\centering
\includegraphics[width=.72\linewidth]{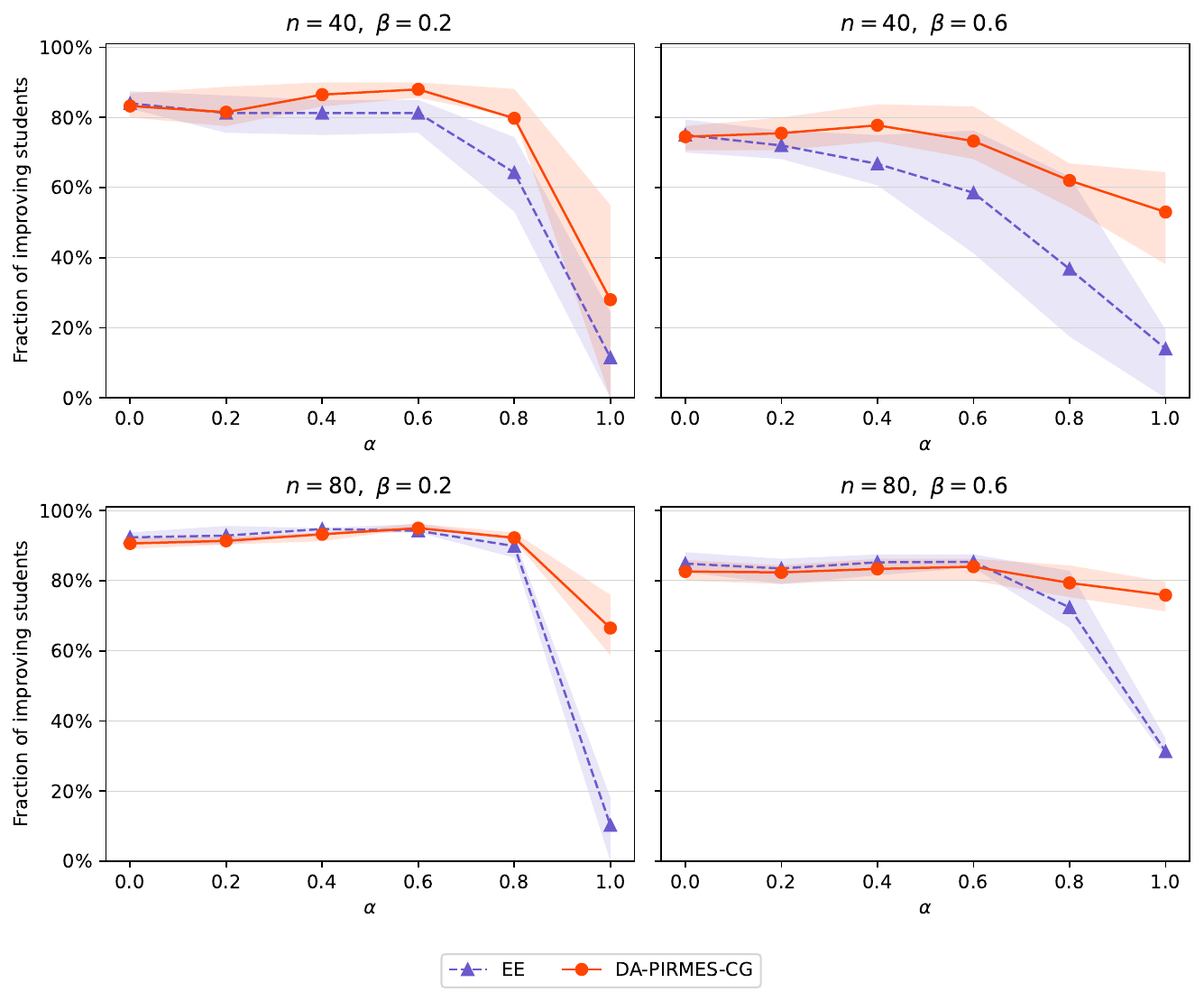}
\caption{Average fraction of improving students upon DA as a function of $\alpha$ for methods EE and DA-\name-CG, where the shaded areas display the interquartile ranges (25\%-75\%).}
\label{fig:fraction_impr_DA}

\end{figure}
In this section, we replicate the computational experiments in the seminal paper by \citet{erdil2008s}, and evaluate how our proposed solution concept DA-\name-CG performs with respect to their method EE. We evaluate both methods on generated instances with [40 students and 8 schools] and [80 students and 16 schools], $\beta\in[0.2, 0.6]$, varying values of $\alpha$, and complete agent preferences.

Figure~\ref{fig:fraction_impr_DA} shows that, for generated data, the fraction of students that improve upon DA is similar for EE and DA-\name-CG when preferences are relatively uncorrelated (i.e., $\alpha$ is low), while substantially more students improve upon DA in DA-\name-CG than in EE when $\alpha$ is high.


With the fraction of improving students upon DA being similar or higher than in EE, Figure~\ref{fig:avg_rank_impr_DA} shows that the expected improvement in rank upon DA for the improving students is substantially higher for DA-\name-CG than for EE. For 80 students, $(\alpha, \beta) = (0.4,0.2)$, for example, on average, improving students can be assigned to a school that is on average 1.24 spots higher on their preference list compared to DA, while EE only realizes an average improvement in rank of 0.75 spots for improving students. Note that, because the warm start of DA-\name-CG is the random matching EE, DA-\name-CG will always have a weakly lower expected rank than EE. In general, optimality of the obtained result of DA-\name-CG was certified in less than 10 minutes for all instances with 40 students, and in all but 17 of the 120 instances with 80 students.

\begin{figure}[!tb]
\centering
\includegraphics[width=.72\linewidth]{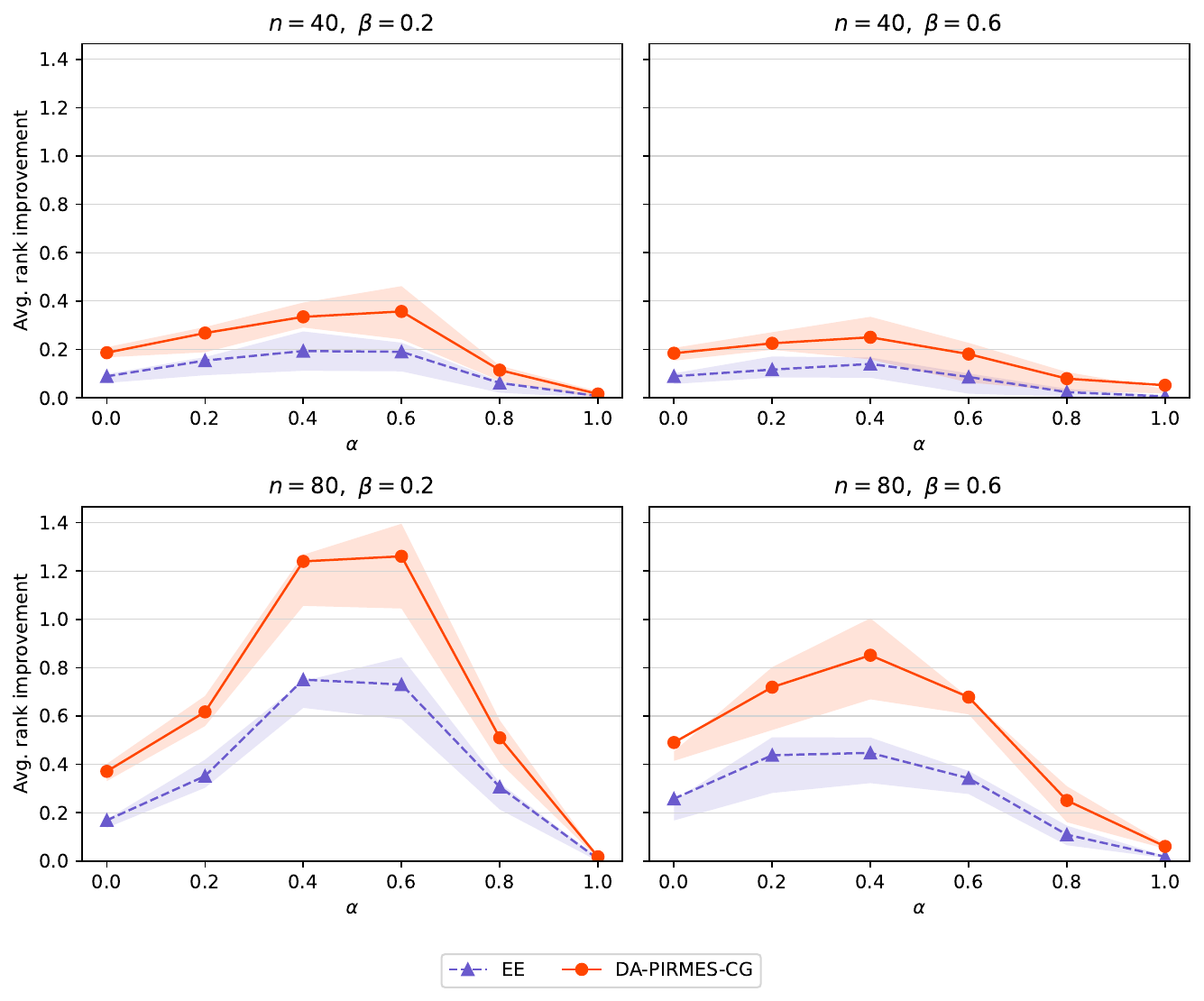}
\caption{Average improvement in rank \textit{among improving students} compared to DA as a function of $\alpha$ for methods EE, and DA-\name-CG, where the shaded areas display the interquartile ranges (25\%-75\%).}
\label{fig:avg_rank_impr_DA}
\end{figure}

Note that our proposed solution concept can be applied to any ex-post stable random matching. In Appendix \ref{app:sd_EE_cg}, we include the computational results of method EE-\name-CG, which sd-improves upon EE. In general, while the fraction of improving students is similar in EE-\name-CG and DA-\name-CG, the average rank improvement in EE-\name-CG is slightly smaller than in DA-\name-CG, because of the reduced solution space. 

\subsection{Comparison with EADA}
\label{subsec:comp_EADA}
We also evaluate the method EADA-\name, which aims to dominate the expected outcome of the Efficiency-Adjusted Deferred Acceptance (EADA) mechanism by \citet{kesten2010school}. In short, EADA allows students to waive their priorities at certain schools whenever this does not affect their own assignment. As a result, the EADA-matching may be unstable. 

In Example \ref{ex:EADA}, we showed that there exist instances where EADA can be sd-dominated by an ex-post stable random matching. This is remarkable, as EADA itself is not guaranteed to be weakly stable when applied to an instance after tie-breaking. Figure \ref{fig:avg_rank_impr_EADA} evaluates how frequently EADA can be sd-dominated by ex-post stable random matchings for generated data by evaluating EADA-\name-CG. The figure displays the improvement in average rank, among improving students, compared to DA, but only for those instances where EADA-\name-CG could identify an ex-post stable random matching sd-dominating EADA. We can see that when preferences are relatively uncorrelated (i.e., $\alpha$ is low), EADA-\name-CG manages to identify an ex-post stable random matching with a lower average rank than EADA. The improvement in expected rank by EADA-\name-CG is lower than for DA-\name-CG, because of the reduced solution space. 

\begin{figure}[!tb]
\centering
\scalebox{0.95}{
\includegraphics[width=.72\linewidth]{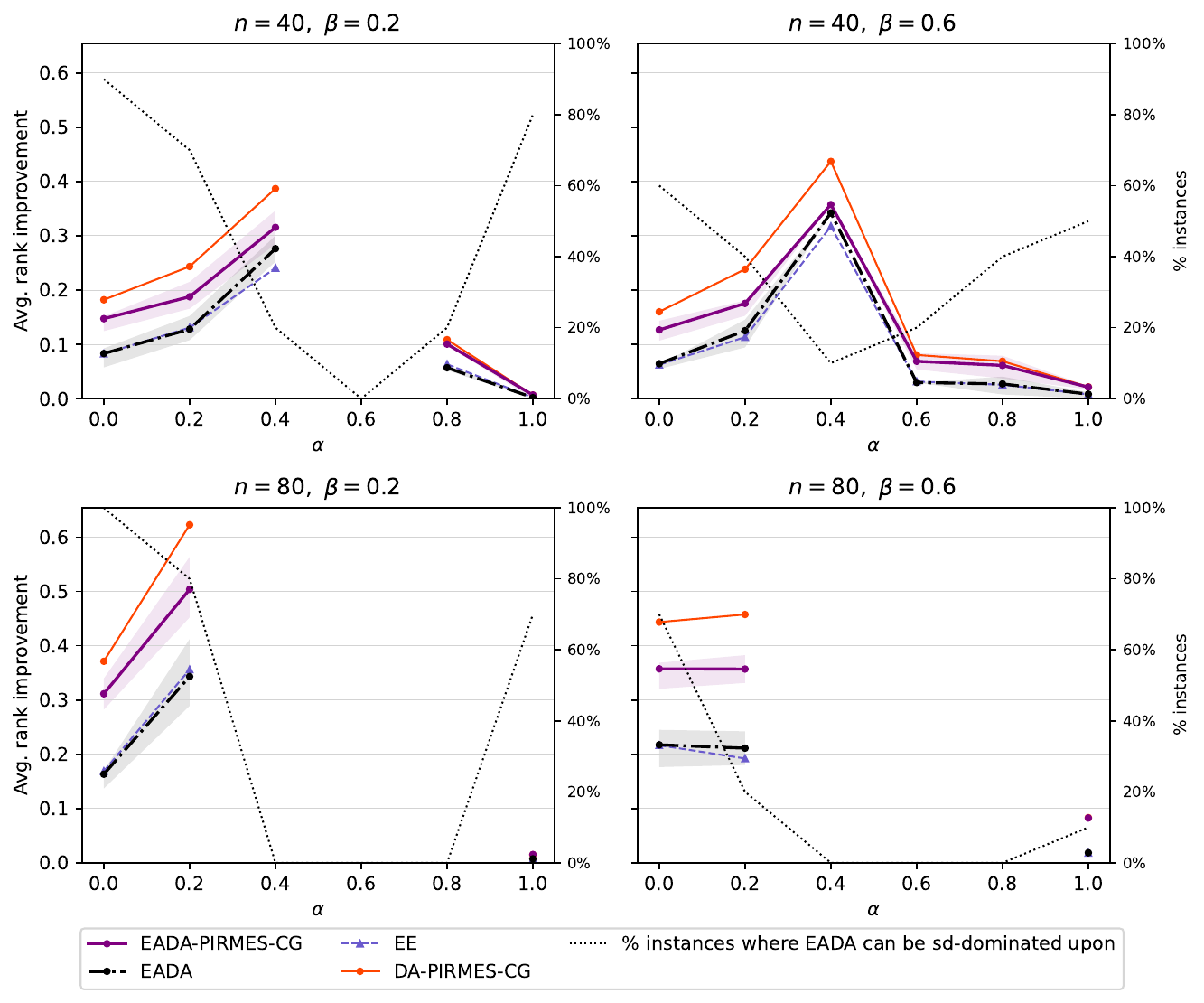}}
\caption{The left axis displays the average improvement in rank \textit{among improving students} compared to DA as a function of $\alpha$, averaged over all instances where EADA could be sd-dominated upon by EADA-\name-CG. The right axis and the dotted line display the fraction of the instances in which EADA could be sd-dominated upon. The shaded areas display the interquartile ranges of EADA-\name-CG and EADA (25\%-75\%).}
\label{fig:avg_rank_impr_EADA}
\end{figure}

\subsection{Estonian Kindergarten Data}
\label{subsec:estonia}
To evaluate the performance of the proposed methods on real-world data, we first consider the 2016 kindergarten allocation in Harku, Estonia \citep{veski2017efficiency, biro2021complexity}. The data contains the strict preferences of 152 families over 7 schools with a total capacity of 155 seats. All but one of the students have full preferences. Moreover, the data contains the distance between the students' homes and the schools, as well as whether the student has a sibling at the school they are proposing to. 

To determine the priorities of the kindergartens, we consider two different methods. First, we denote by {\sc{Sib}} the policy that gives the absolute highest priority to students who have a sibling at a given school. Denote by {\sc{NoSib}} the policy that does not take sibling information into account. Regarding the distance, the {\sc{RelDist}} method gives the same priority to all students who list a school in the same rank in their preferences, and gives higher priority to students who rank a school higher than other students. Alternatively, {\sc{Dist3}} gives the same priority to all students who rank a school in their first three choices, and ranks them above all students who rank that school fourth or worse. Combining these sibling- and distance related policies, we obtain four methods to determine the priorities of the kindergartens: {\sc{Sib-RelDist}}, {\sc{NoSib-RelDist}}, {\sc{Sib-Dist3}}, and {\sc{NoSib-Dist3}}. Note that {\sc{Sib-RelDist}} results in the priorities with the least ties, while {\sc{NoSib-Dist3}} results in the priorities with the most ties.

Table \ref{tab:Estonia} shows the findings of our computational experiments for the Estonian data set. First, both DA-\name-heur and DA-\name-CG improve substantially upon EE: for three of the four priority structures, both the fraction of improving students, as well as their average rank improvement are higher. Consider, for example, priority structure {\sc{NoSib-Dist3}}. While 11\% of the students improve upon DA under EE (in the stochastic dominance sense), DA-\name-CG manages to realize improvements for 74\% of the students. At the same time, the average improvement in rank is also almost twelve times higher in DA-\name-CG, compared to EE. For the fourth priority structure, {\sc{NoSib-RelDist}}, the fraction of improving students doubles, while the average rank improvement \textit{among improving students} is similar in DA-\name-CG and EE. 

Second, we observe that the column generation manages to find certified optimal solutions in relatively limited computation times. For some priority structures, such as {\sc{NoSib-Dist3}}, the optimal average rank improvement among improving students is almost double in comparison to in the solution found by DA-\name-heur. For other priority structures (e.g., {\sc{Sib-RelDist}} or {\sc{NoSib-RelDist}}), the improvement realised by the heuristic is already close to optimal. This illustrates that especially for instances with coarse priorities, such as {\sc{NoSib-Dist3}}, the potential gains of our methods are larger.

Third, EADA obtains random matchings with lower average ranks, but the resulting matchings contained a considerable number of blocking pairs for all priority structures. EADA-\name-heur could not sd-improve upon EADA in any of the four priority structures, thus clearly illustrating the tradeoff between stability and efficiency. 

\begin{table}[!tbp]
\small
    \centering
    \caption{Average rank, fraction of improving students upon DA (\% stud.\ impr.), average rank improvement \textit{among improving students} upon DA (Avg.\ impr.), average number of blocking pairs (\# BP), and computation time (in s) for different methods in Estonian kindergartens data, for different priority structures.}
    \resizebox{\linewidth}{!}{
    \begin{NiceTabular}{lccccr|ccccr}
	\toprule
	& \multicolumn{5}{c}{{\sc{Sib-RelDist}}} & \multicolumn{5}{c}{{\sc{Sib-Dist3}}}\\
    & \multirow{2}{*}{Avg.\ rank} & \% stud.\   & \multirow{2}{*}{Avg. impr.\ } & \multirow{2}{*}{\# BP} & \multirow{2}{*}{Time}& \multirow{2}{*}{Avg.\ rank} & \% stud.\ & \multirow{2}{*}{Avg. impr.\ } & \multirow{2}{*}{\# BP} & \multirow{2}{*}{Time} \\
    &&impr.\ &&&&&impr.\ &&\\\midrule
	DA & 1.7854 & --& --&0& --& 1.7878 & -- & -- & 0 & --\\
	EE & 1.7776 &10\%& 0.0741&0& --& 1.7877 &3\%& 0.0020 & 0 & --\\
	DA-\name-heur & 1.7566 &24\%&0.1183&0&1.40& 1.7748 &35\%& 0.0373&0&1.65\\
	DA-\name-CG (optimum) & 1.7546 & 26\%& 0.1201&0&3.96& 1.7666 & 39\% & 0.0554&0&53.51\\ \midrule
    EADA & 1.7193 & 28\%&0.2393&30.0& --& 1.7577 &45\%& 0.0673 & 27.7 & --\\
	\midrule
    & \multicolumn{5}{c}{{\sc{NoSib-RelDist}}} & \multicolumn{5}{c}{{\sc{NoSib-Dist3}}}\\
    & \multirow{2}{*}{Avg.\ rank} & \% stud.\   & \multirow{2}{*}{Avg. impr.\ } & \multirow{2}{*}{\# BP} & \multirow{2}{*}{Time}& \multirow{2}{*}{Avg.\ rank} & \% stud.\ & \multirow{2}{*}{Avg. impr.\ } & \multirow{2}{*}{\# BP} & \multirow{2}{*}{Time}\\
    &&impr.\ &&&&&impr.\ &&\\\midrule
	DA & 1.7808 &--&--& 0 & -- & 1.7695 &--&--&0 & --\\
	EE & 1.7684 &11\%& 0.1107& 0 & -- &1.7690 &11\%& 0.0046&0& --\\
	DA-\name-heur & 1.7572 &24\%&0.0972&0&1.53& 1.7517 &66\%&0.0268&0&1.73 \\
	DA-\name-CG (optimum) &  1.7550 & 25\% & 0.1033 & 0&16.37& 1.7293 & 74\%& 0.0531 &0&518.33\\ \midrule
    EADA & 1.6705 &38\%&0.2943&74.3& -- & 1.6956 & 57\%& 0.1291 &73.8 & --\\
	\bottomrule
\end{NiceTabular}}
\label{tab:Estonia}
\end{table}

\subsection{Ghent Secondary Education Data}
\label{subsec:ghent}

\begin{table}[!tbp]
\small
\centering
\caption{Average rank, including and excluding being assigned to outside option $\varnothing$, number of unassigned students, fraction of improving students upon DA (\% stud.\ impr.), and average rank improvement \textit{among improving students} upon DA (Avg.\ impr.) for different methods in school choice data Ghent (Belgium).}
\resizebox{\linewidth}{!}{
\begin{NiceTabular}{lccccc|ccccc}
\CodeBefore
\Body
\toprule
    & \multicolumn{5}{c}{{\sc{Ghent 2022}} ($|N| = 4943, |S| = 74$)} & \multicolumn{5}{c}{{\sc{Ghent 2023}} ($|N| = 5032, |S| = 75$)}\\
    & \multirow{2}{*}{Avg.\ rank} & Avg.\ rank & \multirow{2}{*}{\# unassigned}  & \% stud.\   & \multirow{2}{*}{Avg. impr.\ } & \multirow{2}{*}{Avg.\ rank} & Avg.\ rank & \multirow{2}{*}{\# unassigned}  & \% stud.\   & \multirow{2}{*}{Avg. impr.\ }\\
    &&(excl.\ $\varnothing$) &&impr.\ &&&(excl.\ $\varnothing$) &&impr.\ &\\
    \midrule
    DA & 4.2061 & 1.0434 & 208.44 & -- & -- & 4.4438 & 1.0540 & 224.44 & -- & --\\
    EE & 4.2039 & 1.0412 & 208.44 & 15.4\% & 0.0148 & 4.4411 & 1.0512 & 224.44 & 15.8\% & 0.0173\\
    DA-\name-heur & 4.1926 & 1.0414 & 207.69 & 15.5\% & 0.0872 & 4.4192 & 1.0518 & 222.96 & 22.5\% & 0.1098\\
    DA-\name-CG (1h) & 4.1903 & 1.0413 & 207.54 & 18.0\% & 0.0881 & 4.4159 & 1.0519 & 222.73 & 23.3\% & 0.1200\\
    LB after 1h & 3.9912 & -- & -- & -- & -- & 3.8502 & -- & -- & -- & --\\
    \midrule
    EADA & 4.2039 & 1.0412 & 208.44 & 17.6\% & 0.0127 & 4.4412 & 1.0513 & 224.44 & 20.1\% & 0.0132\\
    EADA-\name-heur & 4.1956 & 1.0412 & 207.89 & 20.3\% & 0.0521 & 4.4348 & 1.0513 & 224.02 & 26.6\% & 0.0338\\
    EADA-\name-CG (1h) & 4.1921 & 1.0414 & 207.66 & 20.7\% & 0.0677 & 4.4227 & 1.0517 & 223.19 & 27.5\% & 0.0769\\
	\midrule
    & \multicolumn{5}{c}{{\sc{Ghent 2024}} ($|N| = 5016, |S| = 75$)} & \multicolumn{5}{c}{{\sc{Ghent 2025}} ($|N| = 4701, |S| = 70$)}\\
    & \multirow{2}{*}{Avg.\ rank} & Avg.\ rank & \multirow{2}{*}{\# unassigned}  & \% stud.\   & \multirow{2}{*}{Avg. impr.\ } & \multirow{2}{*}{Avg.\ rank} & Avg.\ rank & \multirow{2}{*}{\# unassigned}  & \% stud.\   & \multirow{2}{*}{Avg. impr.\ }\\
    &&(excl.\ $\varnothing$) &&impr.\ &&&(excl.\ $\varnothing$) &&impr.\ &\\
    \midrule
    DA & 3.9928 & 1.0518 & 194.11 & -- & -- & 5.8020 & 1.0485 & 314.73 & -- & --\\
    EE & 3.9897 & 1.0487 & 194.11 & 14.0\% & 0.0222 & 5.8003 & 1.0468 & 314.73 & 20.7\% & 0.0083\\
    DA-\name-heur & 3.9855 & 1.0478 & 193.89 & 17.8\% & 0.0409 & 5.7888 & 1.0469 & 313.97 & 24.9\% & 0.0529\\
    DA-\name-CG (1h) & 3.9852 & 1.0475 & 193.89 & 18.1\% & 0.0420 & 5.7878 & 1.0469 & 313.91 & 25.7\% & 0.0550\\
    LB after 1h & 3.9581 & -- & -- & -- & -- & 5.5480 & -- & -- & -- & --\\
    \midrule
    EADA & 3.9897 & 1.0486 & 194.11 & 14.7\% & 0.0212 & 5.8003 & 1.0468 & 314.73 & 24.1\% & 0.0071\\
    EADA-\name-heur & 3.9876 & 1.0485 & 193.98 & 20.0\% & 0.0263 & 5.7907 & 1.0469 & 314.10 & 31.9\% & 0.0353\\
    EADA-\name-CG (1h) & 3.9859 & 1.0481 & 193.89 & 21.4\% & 0.0326 & 5.7897 & 1.0469 & 314.03 & 31.9\% & 0.0386\\
	\bottomrule
\end{NiceTabular}}
\label{tab:Ghent}
\end{table}
To further assess the potential of {DA-\name} on real-world instances, we also evaluate our method on school choice data from applications for secondary education in the city of Ghent (Belgium) in the years 2022-2025. The mechanism that is currently used in Ghent is EADA. Given the size of the program (roughly 5,500 students and 75 schools each year), DA and EADA are evaluated based on 10,000 random orderings of the students. School priorities are coarse in this application, with only four priority classes, namely (in decreasing order): parent works at school and sibling currently attends school, parent works at school, sibling currently attends school, others. Appendix \ref{app:descriptive_Ghent} contains descriptive statistics of the four data sets. The 2022 data was already used by \citet{DGHL23} in the context of minimizing the worst-case number of unassigned students by any of the matchings in the lottery of a given assignment (without allowing the possibility for sd-improvements). 

Table \ref{tab:Ghent} shows the performance of the different methods for the different years.\footnote{Note that the size difference in average rank of DA across years is mainly caused by the differing numbers of unassigned students, as this difference disappears when observing the average rank when excluding the outside option (for which the rank is set to $|S|+1$ for every student).} A first observation is that our method {EADA-\name-CG} sd-dominates the currently used EADA assignment in all four years. Between 10-18\% of the students realize a strict sd-improvement upon EADA. The right panel of Figure \ref{fig:GhentDAEADA2023} illustrates how the size of these improvements are distributed for the 2023 data (Appendix \ref{app:Ghent_all} contains the equivalent figures for the other years). Among the 14.7\% improving students in 2023 (739 students), 5.6\% (281 students) experience an improvement of at least 0.1 in their expected rank. We believe these results to be striking, because {EADA-\name-CG} realizes larger efficiency gains upon DA than EADA, while recovering the ex-post stability that was sacrificed by EADA to obtain these efficiency gains. 

With respect to DA, Table \ref{tab:Ghent} also shows that, for all years, {DA-\name} manages to achieve a further decrease in average rank upon DA, compared to EE. In all years, both the fraction of students with strict sd-improvements upon DA, as well as their average rank improvement are larger under {DA-\name-CG} than under EE. In 2023, for example, this translates in an additional  $7.5\%$ of all students (375 students) experiencing an improvement in rank under {DA-\name-CG} after one hour of computation time, compared to EE. The distribution of these rank improvements across students is illustrated in the left panel of Figure \ref{fig:GhentDAEADA2023}. Observe also that the simple {DA-\name-heur} already manages to improve substantially upon EE. 

Lastly, Table \ref{tab:Ghent} also contains information about the maximum possible rank improvements. This lower bound that was described in Section \ref{sec:column_gen} holds for {DA-\name-CG}, and hence also for {EADA-\name-CG}. Although these bounds are not tight, they leave open the possibility for further decreases in average ranks. For 2023, for example, the difference between the currently best solution and this lower bound is 0.57. Further optimization of the implementation of the column generation framework might capture part of this potential improvement, or could tighten the gap between the found solution and the lower bound. The large literature on acceleration and stabilization techniques for column generation can serve as a starting point for future research.

\begin{figure}[!tb]
\centering
\includegraphics[width=.72\linewidth]{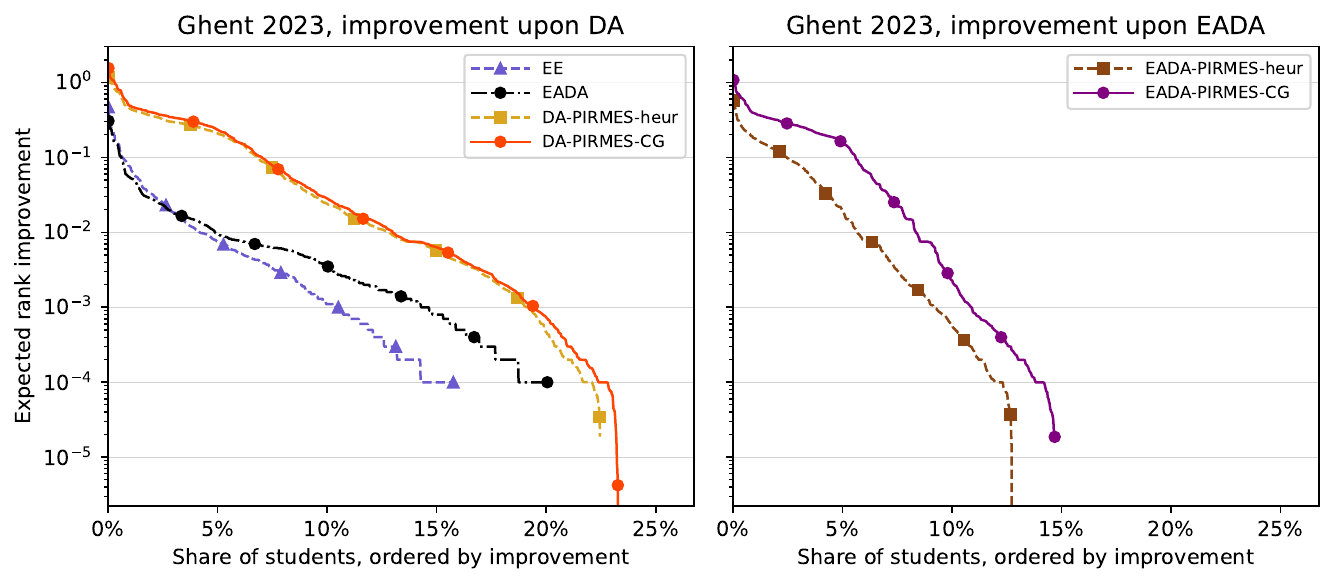}
\caption{Expected individual rank improvements of the students upon DA (left) and upon EADA (right), for different methods in school choice data Ghent (2023).}
\label{fig:GhentDAEADA2023}
\end{figure}

\section{Conclusions}

From a theoretical perspective, we show that testing whether a random matching is constrained-sd-efficient is coNP-complete. This is in stark contrast with the polynomial complexity of testing constrained-efficiency of deterministic matchings, and the polynomial complexity of testing sd-efficiency of random matchings. 

From a practical perspective, we propose the \textit{\namelong} (\name) mechanism, which sd-dominates the expected outcome of a given random matching (e.g., DA with random tie-breaking) without sacrificing ex-post stability. By using advanced optimization techniques such as column generation, we illustrate how this mechanism, applied to DA with random tie-breaking, can substantially reduce the average rank in practical school choice instances in comparison to standard methods to resolve Pareto-inefficiencies, such as the \citeauthor{erdil2008s} method which resolves inefficiencies after ties have already been broken.

\paragraph{Future Directions}

Inspired by the size of the realized efficiency gains, it would be interesting to explore the adaptation of smart lotteries for other applications where inefficiencies exist due to random tie-breaking, such as college admission and resident allocation with ties in the scores, course allocation with coarse priorities, etc. We identify three main research directions towards a general framework to facilitate the adaptation of smart lotteries across these applications.

First, the proposed smart lottery framework is extremely flexible, as it can be applied to improve upon other random matchings than DA, use estimated utilities rather than sd-dominance, optimize objective functions other than minimizing the average rank, or impose different ex-post constraints. Indeed, various variants of smart lotteries have been devised for different applications in the past \citep{bronfman2018redesigning, ashlagi2014improving, DGHL23}, each with their own subtle differences in these criteria. A unified taxonomy of these smart lottery problems, inspired, for example, by the existing taxonomy of scheduling problems, could structure this growing landscape of problems. Such a taxonomy should include the random matching that is being improved upon (e.g., RSD, DA, $\ldots$), whether the random matching is simply decomposed or is also improved upon, and, if so, the objective function (e.g., minimize average rank, maximize neighborhood cohesion, etc.) and the dominance criterion (sd-dominance, estimated utilities) that are used for this improvement, and, lastly, the desired ex-post constraints (e.g., ex-post stability, ex-post Pareto-efficiency, $\ldots$). Other ex-ante stability concepts could also be considered~\citep{AzKl19b}.


Second, it would be interesting to evaluate the effect of combining different objectives when designing smart lotteries. In kidney exchange programs, for example, it is common for policies to optimize multiple objective functions in a hierarchical or weighted way. Can we, for example, simultaneously obtain improvements in welfare with ex-post stability, while also minimizing the worst-case number of unassigned students \citep{DGHL23}, and/or maximizing neighborhood cohesion \citep{ashlagi2014improving}?

Third, the Israeli resident match \citep{bronfman2018redesigning} illustrated that decision-makers are open to implement smart lotteries if ex-ante improvements can be realized. From a pragmatic point of view, how can smart lotteries best be explained to decision-makers in order to facilitate adoption in other settings?

\setlength{\parskip}{0pt}

\normalsize


%
\bibliographystyle{abbrvnat}

\bibliography{abb,haris_master,aziz_personal,main}

\newpage

\clearpage
\begin{appendix}

\section{Missing Proofs from  ~Section~\ref{sec:theory}}
\label{app:proofs}


For some of our NP-completeness reductions, we use the following NP-complete problem.

\begin{problembox}[\XC]
\probleminput{A set $\mathcal{X} = \{ a_1, \dots, a_{3n} \}$ of $3n$ elements and $3n$ sets $\mathcal{C} = \{ C_1, \dots, C_{3n} \}$, each containing 3 elements,  such that each element is covered by exactly 3 sets.}

\problemquestion{Is there an exact 3-cover, that is, a subset $\mathcal{C}' \subset \mathcal{C}$ such that every element is contained in exactly one of them?}
\end{problembox}   
 The NP-completeness of this restricted version of \textsc{exact-3-cover} was first shown by \cite{hein1996complexity}, but only stated explicitly in \cite{hickey2008spr}. 

\subsection*{Proof of Theorem~\ref{thm:improve-NPh}}
\begin{proof}
For containment in coNP, it is easy to see that if $p$ is not constrained-efficient, then there is an ex-post stable matching $q$, which SD-dominates $p$. Hence, given $q$, and a decomposition of $q$ into weakly stable matchings (which constitutes a polynomial size witness), we can easily check that (i) $q$ is indeed ex-post stable, and (ii) $q$ SD-dominates $p$.

To show coNP-hardness, we reduce from the NP-complete \XC\ problem. Let $I=(\mathcal{X},\mathcal{C})$ be an instance of \XC.
We create an instance $I'=(G',\capac, \succsim)$ of \myprob\ as follows.

\begin{itemize}
    \item[--]
For each element $a_i\in \mathcal{X}$, we create a gadget $A_i$, consisting of students $\alpha_i^1,\alpha_i^2,\alpha_i^3$ (they will be referred to as ``element agents'') and schools $\beta_i^1, \beta_i^2$. 

\item[--] For each set $C_j\in \mathcal{C}$, we create a gadget consisting of students $x_j^{\ell},v_j^{\ell},z_j^{\ell}, f_j^{\ell}$ and $y_j^{\ell}$ along with schools $c_j^{\ell},e_j^{\ell},u_j^{\ell}$ and $d_j^{\ell}$, for $\ell \in [3]$. Schools $c_j^{\ell}$ will be referred to as ``set agents''.

\item[--] We create two additional students $b_1,b_2$ and two additional schools $b_1',b_2'$.

\end{itemize}

Every school has capacity 1.
Let $c_l(a_i)$ denote the set agent $c_j^{\ell_i}$, such that $C_j$ is the $l$-th smallest indexed set containing $a_i$ and $a_i$ is the $\ell_i$-th smallest index element in $C_j$. That is, if $C_j=\{ a_1,a_2,a_3\}$ and it is the smallest indexed set containing  $a_2$, then $c_1(a_i)=c_j^{2}$.

Similarly, let $\alpha (c_j^{\ell})$ be the element agent $\alpha_i^l$, such that $a_i$ is the $\ell$-th smallest index element in $C_j$ and $C_j$ is the $l$-th smallest index set containing $a_i$. Finally, let $Y=\{ y_j^{\ell}\mid j\in [3n],\ell \in [3]\}$.
The preferences and priorities are described in Table \ref{tab:prefs}.

\begin{table}[h]
    \centering
    \begin{tabular}{ll|ll}

    $\alpha_i^{l}:$  & $\beta_i^1\succ \beta_i^2\succ c_l(a_i)$  & $\beta_i^1,\beta_i^2:$ & $[\alpha_i^1,\alpha_i^2,\alpha_i^3]$ \\
    $x_j^{\ell}:$ & $d_j^{\ell}\succ c_j^{\ell}\succ e_j^{\ell -1}$ & $c_j^{\ell}:$ & $[v_j^{\ell},\alpha(c_j^{\ell}),y_j^{\ell}]\succ x_j^{\ell}$ \\
    $v_j^{\ell}:$ & $u_j^{\ell}\succ c_j^{\ell}\succ e_j^{\ell}$ & $e_j^{\ell}:$ & $[v_j^{\ell},x_j^{\ell +1}]\succ f_j^{\ell}$ \\
    $y_j^{\ell}:$ & $d_j^{\ell}\succ b_1'\succ c_j^{\ell}$ & $d_j^{\ell}:$ & $[y_j^{\ell},x_j^{\ell}]$ \\
    $z_j^{\ell}:$ & $u_j^{\ell}$ & $u_j^{\ell}:$ & $[z_j^{\ell},x_j^{\ell}]$ \\
    $b_1:$ & $b_1'\succ b_2'$ & $b_1':$ & $[b_2,Y]\succ b_1$ \\
    $b_2:$ & $b_2'\succ b_1'$ & $b_2':$ & $b_1\succ b_2$ \\
    $f_j^{\ell}:$ & $e_j^{\ell}$ & & \\
        
    \end{tabular}
   \caption{The preferences and priorities of the students and schools in Theorem~\ref{thm:improve-NPh}. Brackets indicate indifference.}
    \label{tab:prefs}
\end{table}

The initial random matching $p$ is described in Figures~\ref{fig:hardnessconstr} and \ref{fig:elem}.

\begin{figure}[!htp]
    \centering
    \includegraphics[width=0.9\linewidth]{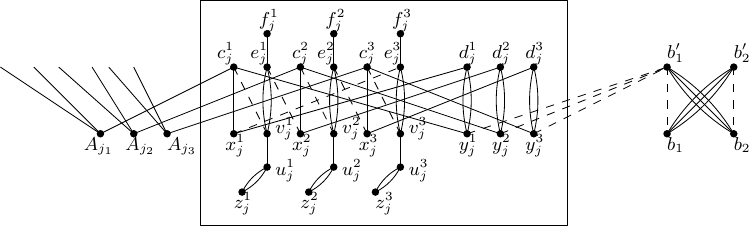}
    \caption{The construction in Theorem~\ref{thm:improve-NPh} for a set $C_j=\{ a_{j_1},a_{j_2},a_{j_3}\}$ with $j_1<j_2<j_3$. The value of each edge is $\frac{1}{3}$ times the number of parallel edges (there is only a single edge in the graph, this is only for visual purposes) for each edge in $p$. The dashed edges have value 0 in $p$.}
    \label{fig:hardnessconstr}
\end{figure}

\begin{figure}[!htp]
    \centering
    \includegraphics[height=0.2\textheight]{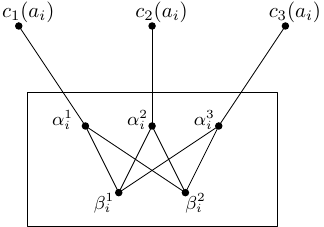}
    \caption{The gadget $A_i$ of an element $a_i$. Every edge has probability $\frac{1}{3}$.}
    \label{fig:elem}
\end{figure}

\begin{figure}[!htp]
    \centering
    \includegraphics[width=0.9\linewidth]{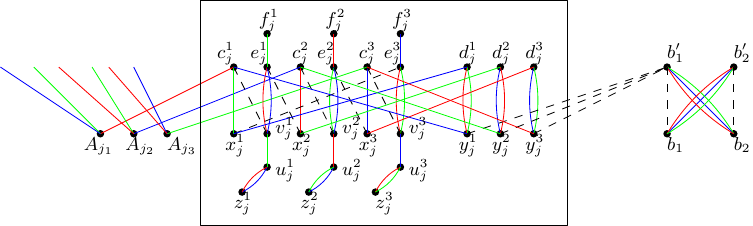}
    \caption{The decomposition of $p$ into three weakly stable matchings, shown by different colors. If $C_j$ was assigned to $a_i$ in matching $\mu_k$ ($k\in [3]$) in $G_I$, where $a_i$ is the $\ell$-th smallest index element of $C_j$, then in $M_k$, $c_j^{\ell}$ is the one from $\{ c_j^1,c_j^2,c_j^3\}$ that gets matched to an element agent, in particular to $\alpha (c_j^{\ell})$.}
    \label{fig:decompof-p}
\end{figure}
\begin{figure}[!htp]
    \centering
    \includegraphics[height=0.2\textheight]{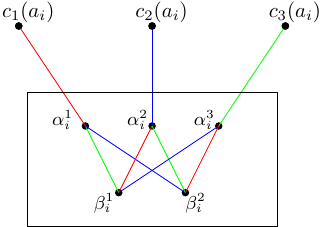}
    \caption{The decomposition of $p$ into three weakly stable matchings, shown by different colors. If $a_i$ was assigned to $C_j$ in matching $\mu_k$ ($k\in [3]$) in $G_I$, where $C_j$ is the $l$-th smallest index set containing $a_i$, then in $M_k$, $\alpha_i^l$ is the one from $\{ \alpha_i^1,\alpha_i^2,\alpha_i^3\}$ that gets matched to a set agent, in particular to $\alpha (c_j^{\ell})$.}
    \label{fig:elem-decompofp}
\end{figure}

\begin{claim}
   The random matching $p$ is ex-post stable.
  
\end{claim}
\begin{proof}
First of all, as every element is contained in exactly three sets in $I$, the bipartite graph $G_I(\mathcal{X},\mathcal{C},E_I)$, where the vertices are the sets and elements respectively, and edges represent the inclusion relations, is 3-regular. Therefore, it is a union of 3 perfect matchings $\mu_1,\mu_2,\mu_3$--a well known result due to~\cite{konig1916graphen}. 

Hence, with the help of $\mu_1,\mu_2,\mu_3$, we can define three matchings $M_1,M_2,M_3$ in $G'$ such that in each of them, exactly one of $c_j^1,c_j^2,c_j^3$ gets matched to an element agent. These are described in Figures~\ref{fig:decompof-p} and \ref{fig:elem-decompofp}.

It is straightforward to verify that all three matchings are weakly stable. If $c_j^{\ell}$ does not get a top ranked student (i.e. it gets $x_j^{\ell}$), then all of $v_j^{\ell },\alpha (c_j^{\ell})$ and $y_j^{\ell}$ get a strictly better school. If $e_j^{\ell}$ does not get a top ranked student (so receives $f_j^{\ell}$), then both $x_j^{\ell +1}$ and $v_j^{\ell}$ get a better school. Schools $b_1',b_2'$ always get their best student. Finally, the other schools are completely indifferent and all receive a student.
\end{proof}
The following claim is the central tool of our hardness proof.
\begin{claim}\label{claim:whocanimprove}
    If $q$ is a random matching that SD-dominates $p$, then $p(e)=q(e)$ for any edge $e\notin \{ (b_1,b_1'),(b_1,b_2'),(b_2,b_1'),$  $(b_2,b_2')\}$.
\end{claim}
\begin{proof}
Let $q$ be a random matching that stochastically dominates $p$.
    \begin{enumerate}
        \item For the edges $(\alpha_i^l,\beta_i^1)$ the statement follows from the fact that $\beta_i^1$ is best for all three $\alpha_i^l$ agents, and each gets it with probability $\frac{1}{3}$, so no one can improve it in $q$ without making someone worse in the SD sense. Similarly, as $\beta_i^2$ is the second best for all three, it also holds for the $(\alpha_i^l,\beta_i^2)$ edges.
        \item For edges of type $(\alpha_i^l,c_j^{\ell})$, it follows from point (1), as each student $\alpha_i^l$ must remain matched with probability 1 in $q$, and they have no other acceptable school. 
        \item For edges incident to $d_j^{\ell}$, the statements holds because $d_j^{\ell}$ is saturated in $p$ and every student that considers it acceptable, considers it best.
        \item For edges incident to $u_j^{\ell}$, we get that $q(v_j^{\ell},u_j^{\ell})\ge \frac{2}{3}$, as $u_j^{\ell}$ is the best school of $v_j^{\ell}$ and $q(z_j^{\ell},u_j^{\ell})\ge \frac{1}{3}$ for the same reason. Hence, these edges also must have the same values as $p$.
        \item If $q(y_j^{\ell},b_1')>0$ for some $j\in [3n],\ell\in [3]$, then $b_1$ and $b_2$ could not remain matched with probability 1, as they both only consider $b_1'$ and $b_2'$ acceptable, implying $q(y_j^{\ell},b_1')= p(y_j^{\ell},b_1')=0$. 
        Hence, for the edges $(y_j^{\ell},c_j^{\ell})$ we must have $q(y_j^{\ell},c_j^{\ell})=\frac{1}{3}=p(y_j^{\ell},c_j^{\ell})$, by point (3) and the fact that the $y_j^{\ell}$ students must remain matched with probability 1. 
        \item For the edges incident to $c_j^{\ell}$, we get that $q(x_j^{\ell},c_j^{\ell})=\frac{1}{3}=p(x_j^{\ell},c_j^{\ell})$, because $x_j^{\ell}$ cannot get more probability for $d_j^{\ell}$ by (3) and (5), by $c_j^{\ell}$ being his second best and by $c_j^{\ell}$ already being matched with probability at least $\frac{2}{3}$ from (2) and (5). It also follows that $q(v_j^{\ell},c_j^{\ell})=0=p(v_j^{\ell},c_j^{\ell})$. 
        \item Finally, for the edges incident to $e_j^{\ell}$, we get that $q(v_j^{\ell},e_j^{\ell})\ge \frac{2}{3}$ and $q(f_j^{\ell},e_j^{\ell})\ge \frac{1}{3}$, so the probabilities must remain the same on these edges as in $p$, as the probability of being assigned cannot decrease for any students. 
    \end{enumerate}
The only remaining edges of $G'$ are the ones in the statement of the claim.
\end{proof}

By Claim~\ref{claim:whocanimprove} we obtain that the only possible way to stochastically dominate $p$ is to increase the probability on the edges $(b_1,b_1'),(b_2,b_2')$ by some number $\varepsilon >0$ and decrease by $\varepsilon$ on the edges $(b_1,b_2'),(b_2,b_1')$. 

\begin{claim}
    If there exists an ex-post stable random matching $q$ that stochastically dominates $p$, then there exists an exact 3-cover in $I$.
\end{claim}
\begin{proof}
As we have already observed, this implies that $q(b_1,b_1')>0$ and $q(e)=p(e)$ for $e\notin  \{ (b_1,b_1'),(b_1,b_2'),$  $(b_2,b_1'),(b_2,b_2')\}$. Hence, if $q$ is ex-post stable, then we must have a weakly stable matching $M$ in the support of $q$ that contains $(b_1,b_1')$.

As $b_1'$ prefers any $y_j^{\ell}\in Y$ to $b_1$, we get that every $y_j^{\ell}$ must be matched to a better partner than $b_1'$, so to $d_j^{\ell}$ in $M$.

Observing the probabilities in the gadgets $A_i$, it is easy to see that in any weakly stable matching in the support of $q$, two of $\alpha_i^l, l\in [3]$ are matched to $\beta_i^1,\beta_i^2$ and one of them to a $c_j^{\ell}$ agent. Hence, for each element $a_i$, exactly one $\alpha_i^l$ element agent is matched to a set agent. 

Suppose that an element agent $\alpha_i^l$ is matched to a set agent $c_j^{\ell}$, but $c_j^{\ell-1}$ is not matched to an element agent. Then, as $c_j^{\ell -1}$ is matched with probability 1, we get that $(x_j^{\ell-1},c_j^{\ell -1})\in M$. As $M$ is weakly stable, but $(c_j^{\ell -1},v_j^{\ell-1})$ does not block $M$, we get that $v_j^{\ell-1}$ must be matched to $u_j^{\ell-1}$ in $M$. Hence, we get that in $M$, $x_j^{\ell}$ must be unmatched and $e_j^{\ell -1}$ must be matched to $f_j^{\ell -1}$ as $q(x_j^{\ell},e_j^{\ell-1})=0$. However, this implies that $(x_j^{\ell},e_j^{\ell-1})$ blocks $M$, a contradiction.

Therefore, we also obtain that for each set $C_j$, if one of its $c_j^{\ell}$ set agents get matched to an element agent $\alpha_i^l$, then all three of them do. Hence, we obtain that there must exist an exact 3-cover in $I$.
\end{proof}

\begin{claim}
    If there exists an exact 3-cover in $I$, then there exists an ex-post stable random matching $q$ stochastically dominating $p$.
\end{claim}
\begin{proof}

    Suppose that there exists an exact 3-cover $\mathcal{C}'\subset \mathcal{C}$. Then, let $q$ be the random matching we get by setting $q(b_1,b_1')=q(b_2,b_2')=\frac{1}{3}$, $q(b_1,b_2')=q(b_2,b_1')=\frac{2}{3}$ and $q(e)=p(e)$ otherwise.
    \begin{figure}[t]
    \centering
    \includegraphics[width=0.9\linewidth]{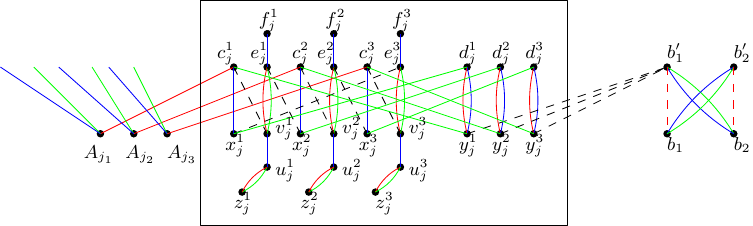}
    \caption{The three matchings in the decomposition of $q$, denoted by different colors for a set $C_j$ \textbf{\textit{within}} the exact 3-cover. The red one corresponds to the exact 3 cover $\mathcal{C}'$, while the blue and green ones are the ones constructed from $E_1$ and $E_2$. }
    \label{fig:incover}
\end{figure}

\begin{figure}
    \centering
    \includegraphics[width=0.9\linewidth]{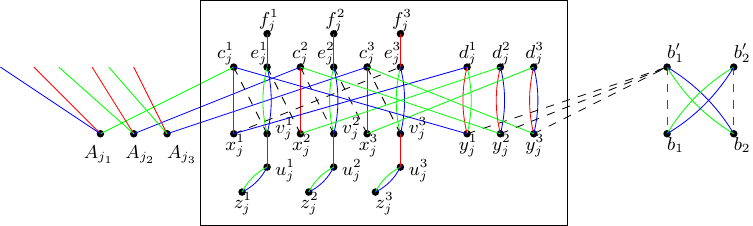}
    \caption{The three matchings in the decomposition of $q$, denoted by different colors, for a set $C_j$ \textbf{\textit{not}} in the exact 3-cover.  The red one corresponds to the exact 3 cover $\mathcal{C}'$, while the blue and green ones are the ones constructed from $E_1$ and $E_2$.}
    \label{fig:notincover}
\end{figure}

    It is easy to verify that $q$ stochastically dominates $p$. 

    As observed before, the bipartite graph $G_I=(\mathcal{X},\mathcal{C},E_I)$, where the vertices are the sets and elements respectively, and edges represent the inclusion relations, is 3-regular, therefore, it is a union of 3 perfect matchings $\mu_1,\mu_2,\mu_3$. If we remove the edges corresponding the assignment of the elements in the exact 3-cover $\mathcal{C}'$ and the vertices of the these sets, we get a bipartite graph $G_I'=(\mathcal{X},\mathcal{C}\setminus \mathcal{C}',E_I')$, where each element-vertex has degree 2 and each set-vertex has degree 3.

    We claim that in this graph $G_I'$, there is an edge set such that every element-vertex is incident to exactly one edge and every set-vertex is incident to one or two edges. To see this, create two copies of every set-vertex $\mathcal{C}\setminus \mathcal{C}'$, and connect them to the same element-vertices. We show that Hall's condition is satisfied for $\mathcal{X}$. For any subset $\emptyset\ne X\subset \mathcal{X}$, there are $4|X|$ incident edges to $X$ now. All of these edges are incident to some vertex in the neighborhood $N(X)$. Suppose for the contrary that $|N(X)|<|X|$. As at most $3|N(X)|<3|X|<4|X|$ edges can be incident to $N(X)$, this is a contradiction. 

    Hence, there is a matching that covers $\mathcal{X}$, so every element-vertex.
    By the existence of $\mu_1$, there is also a matching that covers at least one copy of every set-vertex in $\mathcal{C}\setminus \mathcal{C}'$.
    Hence, there is a matching that covers $\mathcal{X}$ and also at least one vertex of any set by the Mendelsohn-Dulmage theorem~\citep{mendelsohn1958some} (which states that for a bipartite graph $(A,B,E)$, if $X\subseteq A$ can be covered by a matching and $Y\subseteq B$ by another matching, then so can $X\cup Y$). 
    
    Thus, such a matching gives a set of edges $E_1$ in $G_I'$ such that every element-vertex is incident exactly one edge and every set-vertex is incident to one or two edges. 
    Hence, after removing the edges of this matching from $G_I'$ too, we are left with an edge-set $E_2$ that also satisfies that every element-vertex is covered and every set-vertex has one or two incident edges.

    We use these two edge sets $E_1,E_2$, to define $3$ matchings $M_1,M_2,M_3$, illustrated with different colors (red, blue and green) for a set in the exact 3-cover in Figures~\ref{fig:incover} and for a set not in the exact 3-cover in Figure \ref{fig:notincover}. In $M_1$, the set agents corresponding to the exact 3-cover all get their corresponding element agents and none of them does in $M_2$ and $M_3$. In $M_2$ and $M_3$, for each set not in the exact 3-cover, one or two of its set agents get an element agent as defined by $E_1$ and $E_2$ and none of them does in $M_1$. (If it is not $c_j^1$ who is the only set agent obtaining an element agent in $M_2$ or $M_3$, then we cyclically permute the edges to get $M_2$ and $M_3$ - the construction is symmetric).

    In the $A_i$ gadgets, if $\alpha_i^l$ is the one matched out in $M_j$ to a set agent, then $\alpha_i^{l+1}$ is matched to $\beta_i^1$ and $\alpha_i^{l-1}$ is matched to $\beta_i^2$ in $M_j$, the same way as in Figure~\ref{fig:elem-decompofp}. 

    It is easy to verify that $q=\frac{1}{3}(M_1+M_2+M_3)$. 
    
    Also, it is easy to see that all three are weakly stable. Since most schools are indifferent among most students, the only possibilities for a blocking is when $e_j^{\ell}$ is matched to $f_j^{\ell}$, but $x_j^{\ell+1}$ is unmatched; or when $x_j^{\ell}$ is matched to $c_j^{\ell}$ and $v_j^{\ell}$ is matched to $e_j^{\ell}$, none of which happen in any of $M_1,M_2,M_3$. Finally, edges incident to $b_1'$ and $b_2'$ also do not block in any of the matchings, because when $b_1'$ obtains $b_1$ in $M_1$, then each better student ($b_2$ and $Y$) get matched to their top choice. 
\end{proof}
The coNP-completeness now follows from the above claims.
\end{proof}


\subsection*{Proof of Theorem~\ref{th:bestimproveEE}}
\begin{proof}
Containment in NP is trivial, as it is easy to check if a matching $M'$ is weakly stable and Pareto-improves the students from $M$.

    To show NP-hardness, we reduce from \comsmti. Let $I=(U,W;E;(\succ_u)_{u\in U},(\succsim_w)_{w\in W})$ be an instance of \comsmti\ with $|U|=n$. We further assume that every preference list is at most 3 long. This version remains NP-complete~\citep{Manlove_etal2002}. Also, assume $n\ge 4$ without loss of generality.

    We create an instance $I'$ of \bestimproveEE\ as follows. 
    \begin{itemize}
       
    \item We have students $i$ and $i'$ for all $u_i\in U$.
    \item We have schools $s_j,s_j'$ for all $w_j\in W$ with capacity 1.

    \item We have $n^2$ dummy students $d_1,\dots, d_{n^2}$ and dummy schools $t_1,\dots, t_{n^2}$.

\end{itemize}

Hence, we have $N=\{ i,i'\mid i\in [n]\} \cup \{ d_i\mid i\in [n^2]\}$ and $S=\{ s_j,s_j'\mid j\in [n]\}\cup \{ t_j\mid j\in [n^2]\}$.

Next, we describe the preference profile in $I'$.

\begin{itemize}
    \item 
The preferences of a student $i$ is obtained from $\succ_{u_i}$ by substituting $w_j$ with $s_j$, then appending $t_1\succ t_2\succ \cdots \succ t_{n^2}\succ s_i'$.
     \item The preferences of a student $i'$ is $s_i'\succ s_i$.
     \item For a school $s_j$, the priority list is obtained from $\succ_{w_j}$ by substituting $u_i$ with $i$ with the addition that student $j'$ is strictly preferred to all of them. 
     \item For a school $s_j'$, it has priority list $j\sim j'$.
    \item A dummy student $d_i$ only finds $t_i$ acceptable and $t_i$ has priority list $d_i\succ 1\succ 2\succ \cdots \succ n$. 
    \end{itemize}

Let $R=\frac{n^3}{n^2+2n}$ and let the initial weakly stable matching be $M=\{ (i,s_i'),(i',s_i)\mid i\in [n]\} \cup \{ (d_i,t_i)\mid i\in [n^2] \}$. $M$ is weakly stable, as every school obtains a student it ranks highest. 
We show that there exists a stable matching $M'$ Pareto-improving $M$ for the students with $\avgrank (M)-\avgrank (M')\ge R$ if and only if $I$ admits a complete weakly stable matching.

For the first direction, suppose there is a complete stable matching $\mu$ in $I$. Let $\mu'=\mu\cup \{ (i',s_i')\mid i\in [n]\}\cup \{ (d_i,t_i\mid i\in [n^2]\}$. As $\mu$ was stable, it is easy to see that so is $\mu'$. It is also easy to see that $\mu'$ is a Pareto improvement for the students in $I'$. Furthermore, the rank of the match for each student $i$ for $i\in [n]$ decreases by at least $n^2$, so the average rank also decreases by at least $\frac{n^3}{n^2+2n}$ as needed.

For the other direction, suppose that there is a stable matching $M'$ in $I'$ that Pareto improves the students and $\avgrank (M)-\avgrank (M')\ge R$. Then, the sum of the ranks in $M'$ must be at least $n^3$ smaller than in $M$. We show that $M'$ restricted to $\{ 1,\dots, n\} \cup \{ s_1,\dots, s_n\}$ corresponds to a complete stable matching in $I$. First, suppose that some student $i$ does not get a school $s_j$ in $M'$. By stability of $M'$, $(d_i,t_i)\in M'$ for $i\in [n^2]$. Hence, using that the preference lists of $u_i\in U$ have length at most 3, the sum of ranks can only decrease by at most $n + (n-1)(n^2+3)\le n^3-n^2+4n -3< n^3$ as $n\ge 4$ (the first term is the maximum decrease for students $i', i\in [n]$ who can each decrease their rank by at most 1). This shows that $M'$ must give a matching in $I$ that is complete. Weak stability also follows from the weak stability of $M'$ and the correspondence between the preferences in $I$ and $I'$.
\end{proof}

\section{An Integer Programming Formulation}
\label{sec:IP}

In this section we provide an integer programming formulation to decide if an (ex-post stable) random matching $p$ is constrained-efficient, and if not, find an ex-post stable random matching $q$ that minimizes the average rank of the students among all random matchings that sd-dominate $p$. First of all, by adding a dummy school $\varnothing$ with capacity $c[\varnothing] = |N|$ that is acceptable to every student, but is strictly worse than any other acceptable schools, we can assume that any weakly stable matching is student-perfect, so any random matching that is ex-post stable assigns every student with probability 1. 

To enforce that $q$ sd-dominates $p$, and is a lottery over feasible matchings and has minimal expected rank, we include the following constraints


\begin{align}
& \text{min} & \sum_{(i,s_j)\in E} q(i,s_j)\cdot \rank_i(s_j)  & &\label{IP:start}\\
&\text{s.t.}& \sum_{l=1}^{n^2+1} \lambda_l\cdot M_l(i,s_j)& = q(i,s_j) & ((i,s_j) \in E)  \label{IP:equality}\\
&& \sum_{l=1}^{n^2+1}\lambda_l & = 1 &\label{IP:lamdbacons} \\
&&\lambda_l & \geq 0  & (l\in [n^2+1])\label{IP:lamdbanonneg}\\
&& \sum\limits_{s_{j'}: s_{j'}\succsim_i s_j}q(i,s_{j'}) & \ge \sum\limits_{s_{j'}: s_{j'}\succsim_i s_j}p(i,s_{j'}) & (i,s_j)\in E\label{IP:sddom}\\
&& \sum_{s_j\in A(i)} M_l(i,s_j) & =1 & (i\in N,l \in [n^2+1]) \label{IP:perfcons}\\
&& \sum_{i\in A(s_j)}M_l(i,s_j) & \le c[s_j] & (s_j\in S, l \in [n^2+1])\label{IP:capaccons}\\
&&M_l(i,s_j) & \in \{ 0,1\} & ((i,s_j)\in E, l\in [n^2+1])\label{IP:end}\\
&& M_l &\text{ \textit{is weakly stable}}& (l\in[n^2+1])\label{IP:M_stable}
\end{align}

 Constraints \eqref{IP:equality}-\eqref{IP:lamdbanonneg} ensure that $q$ is a convex combination of the $M_l$, $l\in [n^2+1]$ weakly stable matchings. By Carathéodory's theorem, we can assume that $n^2+1$ matchings suffices (following the realistic assumption in school choice that $|S|\leq n$). Constraints \eqref{IP:sddom} ensure that $q$ sd-dominates $p$. Lastly, Constraints \eqref{IP:perfcons}-\eqref{IP:end} ensure that the matchings in the lottery are feasible. To impose that the matchings in the lottery are weakly stable (Constraints (\ref{IP:M_stable})),  we could impose Constraints (\ref{con:cutoff_begin})-(\ref{con:cutoff_end}) based on cut-off ranks for each matching $M_l$, $l\in [n^2+1]$, as described in detail in Section \ref{sec:column_gen}.

Note that Constraints (\ref{IP:equality}) are not linear as both $\lambda_l$ and $M_l(i, s_j)$ are decision variables. One possible way to linearize these constraints is by replacing them with the following set of constraints, which use auxiliary variables $z_l(i,s_j)$ for each matching $M_l$, $l\in [n^2+1]$, and for each $(i,s_j)\in E$.
\begin{align}
    && z_l(i,s_j) &\leq \lambda_l &((i,s_j)\in E, l\in[n^2+1])\label{con:nonLin1}\\
    && z_l(i,s_j) &\leq M_l(i,s_j) &((i,s_j)\in E, l\in[n^2+1])\label{con:nonLin2}\\
    && z_l(i,s_j) &\geq \lambda_l - (1-M_l(i,s_j)) &((i,s_j)\in E, l\in[n^2+1])\label{con:nonLin3}\\
    && z_l(i,s_j) &\geq 0 &((i,s_j)\in E, l\in[n^2+1])\label{con:nonLin4}\\
    &&\sum_{l=1}^{n^2+1} z_l(i,s_j)& = q(i,s_j) & ((i,s_j) \in E)
\end{align}
If $M_l(i,s_j)=1$, then Constraints (\ref{con:nonLin1}) and (\ref{con:nonLin3}) enforce $z_l(i,s_j)= \lambda_l$. If $M_l(i,s_j)=0$, then Constraints (\ref{con:nonLin2}) and (\ref{con:nonLin4}) enforce $z_l(i,s_j)= 0$. 

Observe that this formulation contains a lot of symmetry because it aims to find $n^2+1$ weakly stable matchings simultaneously. Below, we present two alternative ways to reduce some of this symmetry. A first way is to add the following set of $n^2$ constraints:
\begin{align}
    && \sum_{s_j\in A(1)} j\cdot M_l(1,s_j) &\leq \sum_{s_j\in A(1)} j\cdot M_{l+1}(1,s_j) &(l\in[n^2])
\end{align}
By adding these constraints, we rank the matchings by the school to which the first student is assigned.

Alternatively, we could add the following $n^2$ constraints:
\begin{align}
    && \lambda_l &\geq \lambda_{l+1} &(l\in[n^2])
\end{align}
These constraints will rank the matchings in non-increasing order of the weights with which they are selected in the lottery.

\section{Extension: Equal Treatment of Equals}
\label{app:equal_treatment}
The model in Section \ref{sec:column_gen} does not explicitly enforce that students with equal preferences and equal priorities at the schools will receive the same assignment probabilities. To incorporate this minimal fairness notion into our solution, we can extend the column generation in the following way.

To formalize which students we consider to be identical, we restrict our attention to students with identical \emph{relevant} preferences. Given a random matching $p$, let $\succ_i^{+p}$ denote the restricted preference list of agent $i$ which only contains the objects that are more preferred than the least preferred object for which she receives a strictly positive probability in random matching $p$. Formally, given a random matching $p$, a student $i\in N$, and two schools $s_k, s_\ell\in S$, it holds that $s_k \succ_i^{+p} s_\ell$ if and only if $s_k \succ_i s_\ell$, and there exists a school $s_r\in S$ such that $p(i,s_r)>0$ and $s_\ell \succ_i s_r$.

The set of identical student pairs for a given random matching $p$ is modeled by a binary matrix $I^p \in \{0,1\}^{n\times n}$, where $I^p_{i,j} = 1 = I^p_{j,i}$ if and only if student pair $(i,j)$ are identical with respect to relevant preferences. For any student pair $(i,j)$ and random matching $p$, $I^p_{ij} = 1$ if \begin{enumerate}[label = (\alph*)]
    \item students $i$ and $j$ have the same restricted preference list, i.e., $\succ_i^{+p} = \succ_j^{+p}$,
    \item each school in their restricted preference lists is indifferent between students $i$ and $j$, i.e., $i\sim_{s_k} j$ for each $s_k \in \succ_i^{+p}$.
\end{enumerate}
Note that the values of $I^p$ can be determined in polynomial time by iterating over the students. Denote by $\mathcal{I}(p)$ the set of student pairs $(i,j)\in N\times N$ for which $I_{ij}^p = 1$ in random matching $p$. Clearly, requiring that identical students in $\mathcal{I}(p)$ be treated identically is weaker than requiring that students with fully identical preferences and priorities be treated equally.


To ensure equal treatment of students with identical relevant preferences, we add the following constraints to [P($\tM$)], in order to enforce that each student pair in $\mathcal{I}(p)$ are assigned with the same probabilities:
\begin{align}
    \sum_{\ell:M_\ell\in\tilde{\M}}\lambda_\ell\cdot M_\ell(i, s_{k}) &= \sum_{\ell:M_\ell\in\tilde{\M}}\lambda_\ell\cdot M_\ell(j, s_{k}) & ((i,j) \in \mathcal{I}(p),\text{with }i< j, s_k\in \succ^{+p}_i)\label{IP:equal}
\end{align}

Denote by $\eta_{ijk}$ the dual variables of Constraints~(\ref{IP:equal}), with $(i,j)\in \mathcal{I}(p)$, and $s_k\in S$. The modified pricing problem, replacing expressions (\ref{con:pricing_obj}) and (\ref{con:pricing_con}), then becomes:
\begin{align}
    &\text{max} &\sum_{(i,s_k)\in E}\left(\mu_{ik}\left(\sum\limits_{s_{k'}: s_{k'}\succsim_i s_k} \left(M(i,s_{k'})\right) \right)\right.&-M(i,s_k)\cdot \rank_i(s_k)\Biggr) \nonumber\\
    &&+\sum_{(i,j)\in\mathcal{I}(p)} \sum_{s_k\in \succ^{+p}_i}\left(\eta_{ijk} \left(M(i,s_k) \right.\right.&\left.\left.- M(j, s_k)\right)\right)  + \delta \label{con:pricing_obj_equal}\\
    &\text{s.t.}&M&\in\M \label{con:pricing_con_equal}
\end{align}

\begin{remark}
Requiring equal treatment of equals might be hard to implement when the random matching $p$ that we are sd-improving upon itself violates equal treatment of equals. This can happen, for example, when approximating the expected outcome of DA with random tie-breaking by sampling uniformly at random the set of tie-breaking rules that are used (following Theorem \ref{thm:uniform}). Consider an instance where two identical students $i$ and $j$ receive the following probability vectors for schools of decreasing order in the approximation $\hat{x}^{DA}$ of DA: $\hat{x}^{DA}_i = (0.51, 0.49, \ldots)$ and $\hat{x}^{DA}_j = (0.49, 0.51, \ldots)$. Then, there does not exist a random matching that sd-dominates $\hat{x}^{DA}$ while giving the same probabilities to identical agents $i$ and $j$.
\end{remark}





\section{Computational Results EE-\name-CG}
\label{app:sd_EE_cg}
Figures \ref{fig:fraction_impr_EE} and \ref{fig:avg_rank_impr_EE} are equivalent to Figures \ref{fig:fraction_impr_DA} and \ref{fig:avg_rank_impr_DA}, but additionally show the fraction of improving students and their average rank improvement for EE-\name-CG, which sd-improves upon EE instead of upon DA. It can be seen that the fraction of the students who improve upon DA is slightly higher in EE-\name-CG, in comparison to DA-\name-CG. At the same time, the average improvement in rank by EE-\name-CG lies in between the improvements realized by EE and DA-\name-CG.
\begin{figure}[!htb]
\centering
\includegraphics[width=.72\linewidth]{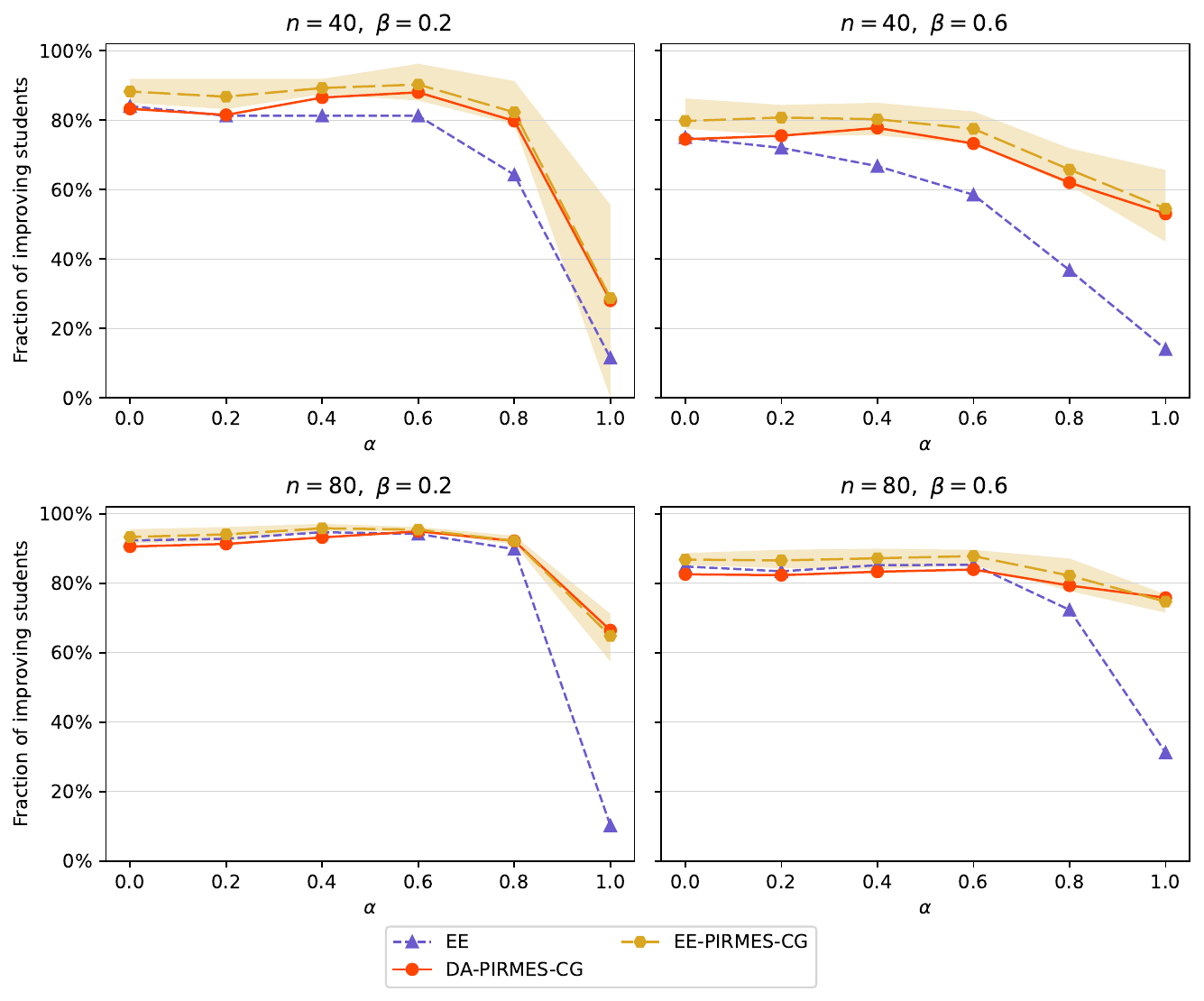}
\caption{Average fraction of improving students upon DA as a function of $\alpha$ for methods EE, DA-\name-CG, and EE-\name-CG, where the shaded areas display the interquartile ranges for EE-\name-CG (25\%-75\%) (see Figure \ref{fig:fraction_impr_DA} for interquartile ranges EE and DA-\name-CG).}
\label{fig:fraction_impr_EE}

\end{figure}

\begin{figure}[!htb]
\centering
\includegraphics[width=.72\linewidth]{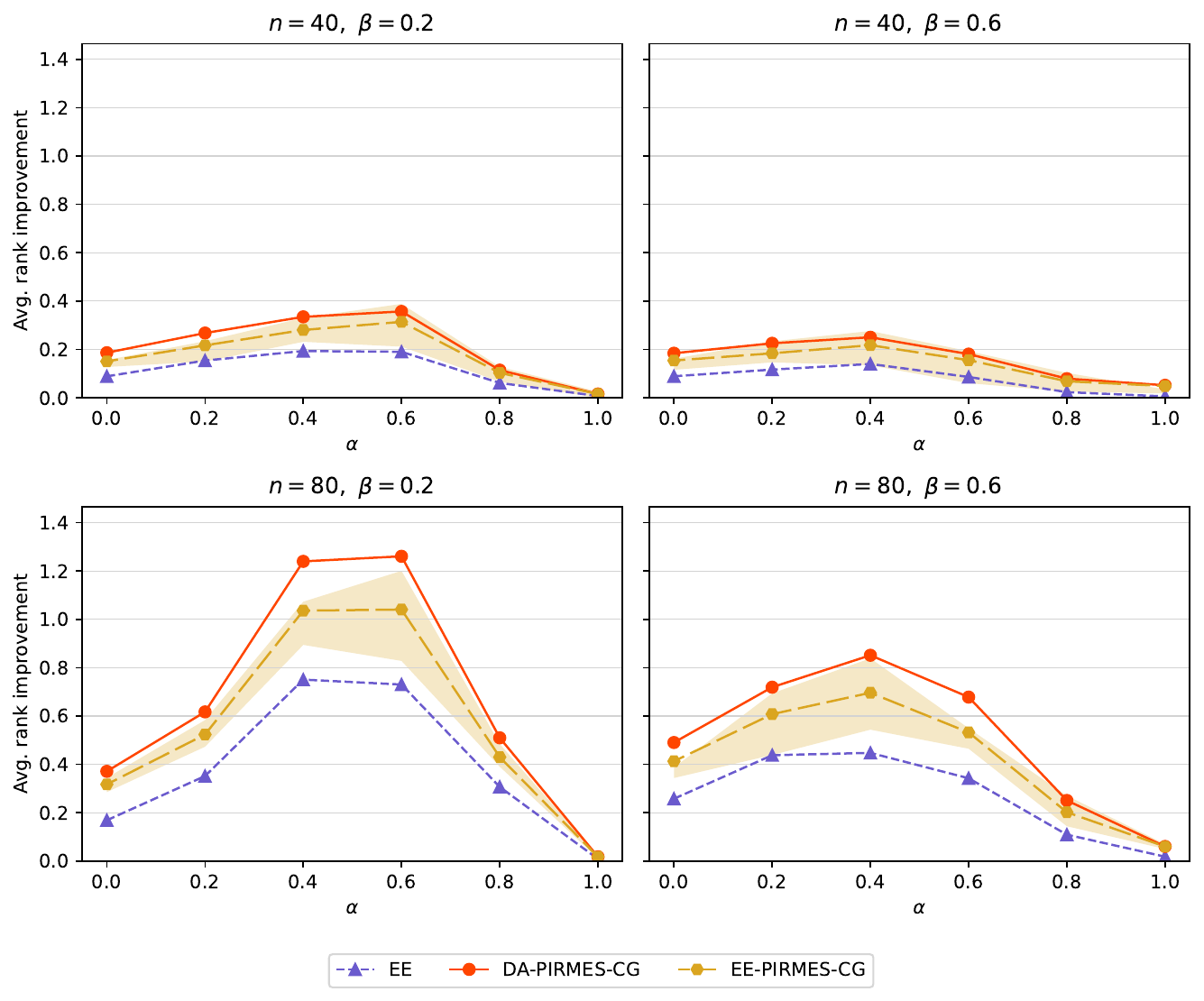}
\caption{Average improvement in rank compared to DA as a function of $\alpha$ for methods EE, DA-\name-CG, and EE-\name-CG, where the shaded areas display the interquartile ranges for EE-\name-CG (25\%-75\%) (see Figure \ref{fig:avg_rank_impr_DA} for interquartile ranges EE and DA-\name-CG).}
\label{fig:avg_rank_impr_EE}
\end{figure}

\newpage
\section{Descriptive Statistics Data Ghent}
\label{app:descriptive_Ghent}
This section contains some descriptive statistics of the application data for secondary education in the Belgian city of Ghent.
\begin{table}[!htbp]
\centering
\caption{Descriptive statistics of the Ghent school choice data (2022-2025).}
\label{tab:ghent-instances}
\begin{tabular}{lrrrr}
\toprule
 & \textbf{2022} & \textbf{2023} & \textbf{2024} & \textbf{2025} \\
\midrule
Students                                    & 4943  & 5032  & 5016  & 4701  \\
Schools                                     & 74    & 75    & 75    & 70    \\
Total capacity                              & 5654  & 5628  & 5770  & 5315  \\
Seats per student                           & 1.144 & 1.118 & 1.150 & 1.131 \\
\midrule
Preference list length, mean                & 2.165 & 2.229 & 2.268 & 2.297 \\
Preference list length, median              & 2     & 2     & 2     & 3     \\
Preference list length, maximum             & 9     & 11    & 12    & 10    \\
Students listing a single school            & 1642  & 1508  & 1478  & 1354  \\
\midrule
Priority groups per school, mean            & 2.378 & 2.573 & 2.493 & 2.329 \\
Schools with a single priority group        & 6     & 5     & 4     & 6     \\
\bottomrule
\end{tabular}
\end{table}

\newpage
\section{Results Data Ghent All Years}
\label{app:Ghent_all}
This section contains the equivalent plots of Figure \ref{fig:GhentDAEADA2023}, which show the distribution of the rank improvements upon DA (left) and EADA (right), for the years 2022, 2024, and 2025, for the application data for secondary education in the Belgian city of Ghent.

\begin{figure}[!hbt]
\centering
\includegraphics[width=.72\linewidth]{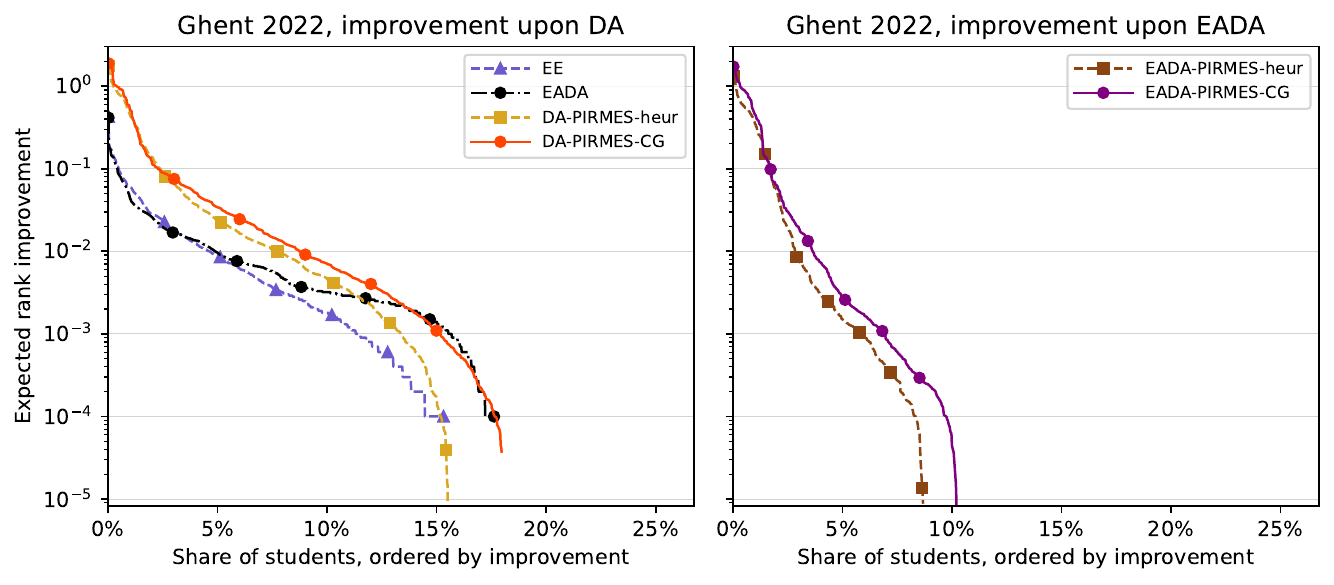}
\caption{Expected individual rank improvements of the students upon DA (left) and upon EADA (right), for different methods in school choice data Ghent (2022).}
\label{fig:GhentDAEADA2022}
\end{figure}

\begin{figure}[!hbt]
\centering
\includegraphics[width=.72\linewidth]{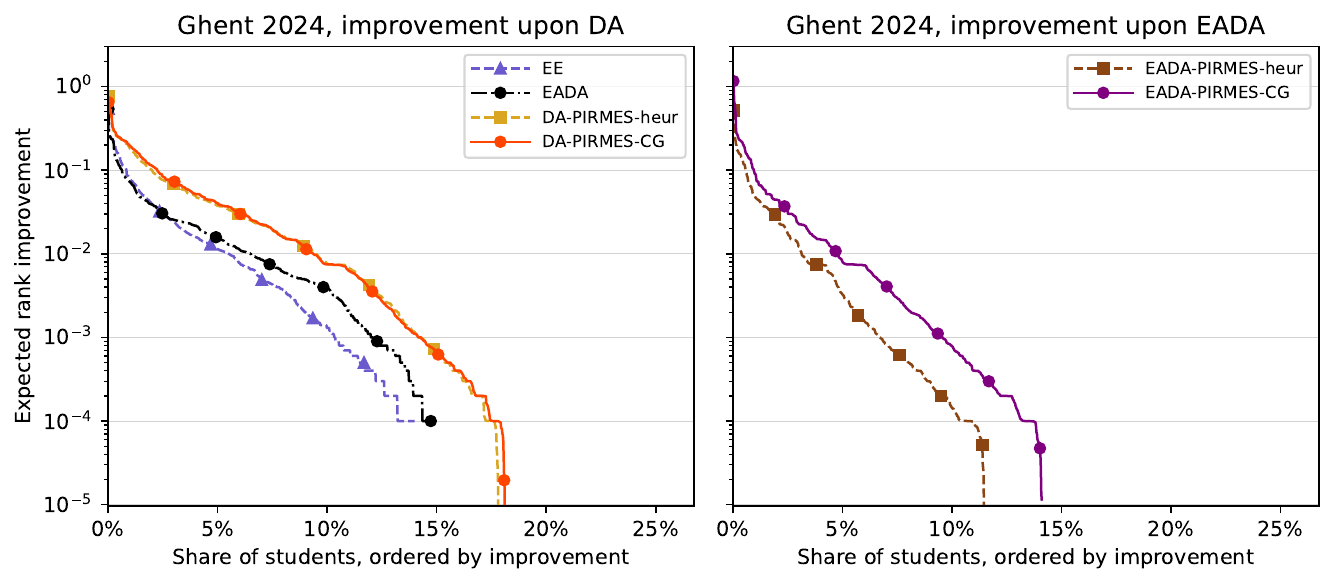}
\caption{Expected individual rank improvements of the students upon DA (left) and upon EADA (right), for different methods in school choice data Ghent (2024).}
\label{fig:GhentDAEADA2024}
\end{figure}

\begin{figure}[!hbt]
\centering
\includegraphics[width=.72\linewidth]{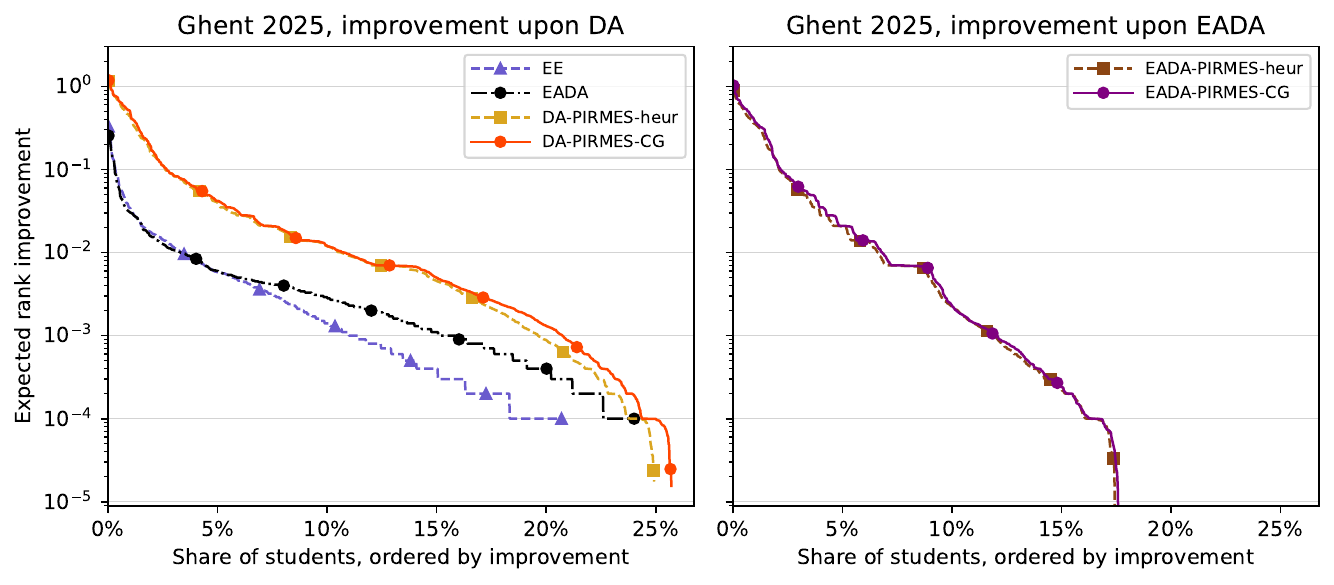}
\caption{Expected individual rank improvements of the students upon DA (left) and upon EADA (right), for different methods in school choice data Ghent (2025).}
\label{fig:GhentDAEADA2025}
\end{figure}

\end{appendix}

\end{document}